\documentclass[journal]{IEEEtran}

\IEEEoverridecommandlockouts
\usepackage{cite}
\usepackage{amsmath,amssymb,amsfonts}
\usepackage{algorithmic}
\usepackage{graphicx}
\usepackage{textcomp}
\usepackage{xcolor}

\usepackage{graphics} 
\usepackage{epsfig} 
\usepackage{mathptmx} 
\usepackage{times} 

\usepackage{amsthm} 
\newtheorem{theorem}{Theorem}
\newtheorem{lemma}[theorem]{Lemma}

\begin{document}
\title{Multi-Polarization Superposition Beamforming:\\Novel Scheme of Transmit Power Allocation and Subcarrier Assignment
\thanks{This work was supported financially by Aerospace Corporation (G255621100) and CSULB Foundation Fund (RS261-00181-10185).}
\thanks{Paul Oh and Sean Kwon are co-first authors with the equal contribution, and Part of this work was presented at IEEE VTC Fall 2020 \cite{Kwon_Oh_VTC20Fall}. }
}

\author{Paul Oh and Sean Kwon{$^\dag$}\\
Department of Electrical Engineering, California State University, Long Beach\\
Email: Paul.Oh@student.csulb.edu;
{$^\dag$}Sean.Kwon@csulb.edu} 


\maketitle

\markboth{IEEE Transactions on Wireless Communications,~Vol.~xx, No.~xx,~Month~Year}{Shell \MakeLowercase{\textit{et al.}}: Bare Demo of IEEEtran.cls for Journals}

\begin{abstract}
The 5th generation (5G) new radio (NR) access technology and the beyond-5G future wireless communication require extremely high data rate and spectrum efficiency. Energy-efficient transmission/reception schemes are also regarded as an important component. The polarization domain has attracted substantial attention in this aspects. This paper is the first to propose \textit{multi-polarization superposition beamforming (MPS-Beamforming)} with cross-polarization discrimination (XPD) and cross-polarization ratio (XPR)-aware transmit power allocation utilizing the 5G NR antenna panel structure. The appropriate orthogonal frequency division multiplexing (OFDM) subcarrier assignment algorithm is also proposed to verify the theoretical schemes via simulations. The detailed theoretical derivation along with comprehensive simulation results illustrate that the proposed novel scheme of MPS-Beamforming is significantly beneficial to the improvement of the performance in terms of the symbol error rate (SER) and signal-to-noise ratio (SNR) gain at the user equipment (UE). For instance, a provided practical wireless channel environment in the simulations exhibits 8 dB SNR gain for $10^{-4}$ SER in a deterministic channel, and 4 dB SNR gain for $10^{-5}$ SER in abundant statistical channel realizations.
\end{abstract}


\begin{IEEEkeywords}
Polarization, multi-polarization superposition beamforoming (MPS-Beamforming), power allocation, subcarrier assignment, 5G/beyond-5G wireless communication.
\end{IEEEkeywords}

\IEEEpeerreviewmaketitle


\section{Introduction}
The 5th generation (5G) wireless communication system has been first commercialized in 2020, and it is expected to be stabilized during the next decade. 5G new radio (NR) access technology is a furnace of almost all the renowned communication theories and technologies. On the other hand, future wireless communication systems such as beyond-5G and 6th generation (6G) will demand far higher channel capacity and spectrum efficiency than 5G. For this reason, the utilization of the polarization domain has recently attracted substantial attention \cite{Kwon_Oh_VTC20Fall, Hanzo_Pol_Hybrid_BF, Pratt_TWC20_PolSK, Zhang_TCom20, Zafari_Sari_TCom17, Kwon_IGESSC20_SM, Nardelli_Ding_Cost_TWC19_MP_NOMA, Paul_Kwon_Molisch_arXiv_AntSel}. The polarization diversity enables the system performance to be substantially enhanced in a variety of communication schemes such as beamforming \cite{Kwon_Oh_VTC20Fall, Hanzo_Pol_Hybrid_BF}; polarization shift keying \cite{Pratt_TWC20_PolSK, Zhang_TCom20}; spatial modulation (SM) \cite{Zhang_TCom20, Zafari_Sari_TCom17, Kwon_IGESSC20_SM}; non-orthogonal multiple access (NOMA) \cite{Nardelli_Ding_Cost_TWC19_MP_NOMA}; and spatial multiplexing with antenna selection \cite{Paul_Kwon_Molisch_arXiv_AntSel}.

The current state-of-the-art and future key technologies significantly take into account utilizing the polarization domain.
Hybrid beamforming can adopt dual-polarization and the associated codebook design to improve the system performance \cite{Hanzo_Pol_Hybrid_BF, Kim_Love10}.
Modulation schemes can be also advanced with the improved BER/SER focusing on different polarization state of the wireless channel, for which the polarization shift keying is introduced and described in a sophisticated manner \cite{Zhang_TCom20, Pratt_TWC20_PolSK}.
Spatial modulation (SM) can also take the benefit of polarization diversity via adopting fixed dual-polarization antenna elements as in \cite{Zafari_Sari_TCom17} or flexible polarization-agile/reconfigurable antenna elements as in \cite{Kwon_IGESSC20_SM}.
Further, it is shown that deploying dual-polarized massive multi-input multi-output (MIMO) systems improve the performance of the conventional non-orthogonal multiple access (NOMA) scheme \cite{Nardelli_Ding_Cost_TWC19_MP_NOMA}. Comprehensive analyses and simulations for the outage probability and outage sum-rate of the dual-polarized massive MIMO-NOMA networks are provided in \cite{Nardelli_Ding_Cost_TWC19_MP_NOMA}.
The significant improvement of channel capacity in the polarization reconfigurable MIMO (PR-MIMO) system with practical polarization-reconfigurable/agile antenna elements and polarization pre/post-coding, is theoretically derived and validated in simulations \cite{Kwon_Molisch_GLOBECOM15, Paul_Kwon_Molisch_arXiv_AntSel}.

Other aspects of interesting research on MIMO system with the polarization diversity have been actively fulfilled \cite{Kwon_Stuber14_PDMA_TWC, Jootar06, Erceg_Paulraj02, White05, Andrews01, Kwon14_GLOBECOM, Shafi_Molisch06}. Utilizing the polarization domain has significant potential to achieve the aforementioned mission for the future wireless communication system, since even 5G NR is not fully utilizing the polarization domain in a systematic manner, e.g., based on full channel state information (CSI) of polarization.
The polarization diversity improves bit/symbol error rate (BER/SER) \cite{Jootar06, Erceg_Paulraj02, White05, Kwon_Stuber14_PDMA_TWC}. Further, it can also increase the channel capacity \cite{Andrews01, Kwon14_GLOBECOM, Shafi_Molisch06}. It is validated by the aforementioned research works that the impact of polarization on the performance of MIMO communication system is noticeable.

It is straightforward that the polarization domain can be combined with other domains such as time and spatial domains \cite{Kwon_Molisch_GLOBECOM15, Paul_Kwon_Molisch_arXiv_AntSel, Jootar06, White05, Erceg_Paulraj02}, although polarization multiplexing without spatial diversity is also promising \cite{Kwon_Stuber14_PDMA_TWC}. The improvement of the SER is achieved via combining spatial and polarization diversity along with maximum ratio combining (MRC) in the scenarios of fast power control and no power control on Nakagami-$m$ fading channels \cite{Jootar06}. The improvement of channel capacity in the MIMO system with polarization-reconfigurable/agile antennas is also precisely described \cite{Kwon_Molisch_GLOBECOM15, Paul_Kwon_Molisch_arXiv_AntSel}.
On the other hand, it is reported that space-time block coding (STBC) with uni-polarization outperforms STBC with dual-polarization in both uncorrrelated/correlated Rayleigh and Ricean fading channels \cite{White05}. As demonstrated in the previously reported literature, the characteristics of the wireless channel substantially affect the system performance in terms of the symbol/bit error rate (SER/BER) \cite{White05, Erceg_Paulraj02}. Deploying dual-polarized antennas in the MIMO system exhibits lower spatial diversity gain but higher spatial multiplexing gain than the MIMO system with uni-polarized antennas \cite{Erceg_Paulraj02}. In particular, in the Ricean fading channel with the high $K$-factor, the dual-polarized MIMO system is highly beneficial for spatial multiplexing, compared to the uni-polarized MIMO system.
Resource allocation schemes also need to take the impact of polarization into account.

It is no doubt that the polarization domain has its unique characteristics distinct from the spatial domain in MIMO
\cite{Kwon_Molisch_GLOBECOM15, Shafi_Molisch06, Kwon_Stuber11_TVT, Kwon_Stuber13_TVT}. The wireless communication system with multi-polarization needs to fully utilize the unique characteristics of polarized channels, in the same fashion that conventional MIMO system exploits spatial diversity of wireless channels.
In other words, the advantage of multi-polarization antenna elements significantly depends on the condition of wireless channels \cite{Jootar06, White05, Erceg_Paulraj02}. The reason is that different conditions of wireless channels cause different degrees of channel depolarization. The electromagnetic plane waves transmitted with a fixed polarization, e.g., $-45^{\rm o}$ polarization at the gNB have both the copolarization ($-45^{\rm o}$) and the cross-polarization ($+45^{\rm o}$) components at the end of the UE; the ratio of those two components is varying depending on the channel environment even for the static UE. This symptom is called channel depolarization, and reported by the empirical and theoretical research \cite{Shafi_Molisch06, Erceg06, Landmann07, Kwon_Stuber11_TVT}. The polarization misalignment between the transmitter (Tx) and the receiver (Rx) degrades the system performance in several aspects including the received signal power.
Comprehensive understanding of polarization/depolarization in wireless communication systems from the measurement and modeling to the theory and novel schemes have been provided in aforementioned prior works.

One of the key features that 5G NR and beyond-5G communication systems support, is the full utilization of MIMO beamforming, where massive MIMO antenna panels will be supported by the base station (BS) or equivalently, the next-generation Node B (gNB) as illustrated in Fig. \ref{fig:System_Model}.
Each antenna panel has multiple collocated dual-polarization antenna elements, where 8-by-8 dual-polarization antenna array or its subarray has been regarded as one of the strong candidates in 5G NR \cite{Roh_5GBeamforming_ComMag14, Nam_FD-MIMO, Nam_FD-MIMO_Feasibility_JSAC17}. Further, the implementation of polarization-reconfigurable/agile antennas has been accomplished in a variety of design approaches, and has shown the feasibility of supporting multi-polarization in MIMO communication systems \cite{qin2017compound, wolosinski20162, babakhani2016frequency, sun2016novel, cai2017continuously, liao2015polarization}.

In addition to the high channel capacity and spectrum efficiency, beyond-5G and 6G communication systems will demand  high energy efficiency as considered in \cite{Kwon_Oh_VTC20Fall, Zhang_TCom20, Zafari_Sari_TCom17, Kwon_IGESSC20_SM}. Consequently, this is the time for leading researchers and scholars to consider novel energy-efficient wireless communication schemes in addition to the current state-of-the-art technology in 5G NR. Based on the agreements achieved so far by the 5G standard society, 5G beamforming is supported by antenna subarray that consists of several spatially separated antenna elements in one column of the antenna panel with the same polarization, whether it is $+45^{\rm o}$ or $-45^{\rm o}$ polarization \cite{Roh_5GBeamforming_ComMag14, Nam_FD-MIMO, Nam_FD-MIMO_Feasibility_JSAC17}. Although the gNB in 5G NR will support dual-beamforming with both $+45^{\rm o}$ and $-45^{\rm o}$ polarization generated by the aforementioned antenna subarray, the current 5G NR design does not consider the impact of dual-beamforming on the polarization of the superimposed received signal at the end of user equipment (UE), which was previously or conventionally called mobile station (MS). Based on this motivation, this paper provides novel schemes along with the comprehensive analysis and simulation results for the impact of multi-polarization superposition beamforming (MPS-Beamforming) on the symbol error rate (SER) and the energy efficiency in terms of the signal-to-noise ratio (SNR) to meet a certain SER at the UE based on the orthogonal frequency division multiplexing (OFDM) system.

Compared with the prior works, this paper has unique
contributions. In particular, comparing to \cite{Kwon_Oh_VTC20Fall}, this paper provides substantially additional contributions such as more theoretical derivation and analysis to reach the closed form, new subcarrier assignment algorithm, and abundant simulation results with both deterministic and statistical channels.  Appreciable contributions of this paper are summarized as
follows:
\begin{itemize}
  \item providing novel scheme of MPS-Beamforming based on cross-polarization discrimination (XPD) and cross-polarization ratio (XPR)-aware transmit power allocation aligned with 5G antenna structure;
  \item illustrating the theoretical derivation of the XPD/XPR-aware transmit power allocation ratio in MPS-beamforming to yield the highest SER performance for the given channel;
  \item theoretical analysis for the impact of combining transmit beams with different polarization and transmit power ratio on the polarization ellipse and its rotation;
  \item proposing a new subcarrier assignment algorithm to take into account MPS-Beamforming scheme in the OFDM system;
  \item comprehensive simulation results and analyses illustrating the remarkable benefit of adopting the novel MPS-Beamforming scheme with XPD/XPR-aware transmit power allocation and subcarrier assignment; and the impact of XPD on the rotation angle of the polarization ellipse for the received signal.
\end{itemize}

The remainder of this paper is as follows.
Section~\ref{sec:System_Model} presents the system model for the energy-efficient MPS-Beamforming. Theoretical description and mathematical derivation for XPD/XPR-aware transmit power allocation, polarization ellipse, feedback-based fine adjustment of polarization are provided in the subsections in Section \ref{sec:MPS-Beamforming}. The appropriate subcarrier assignment algorithm for the MPS-Beamforming OFDM system is also presented in a subsection of Section \ref{sec:MPS-Beamforming}. Further, Section~\ref{sec:Simulation_Results} provides abundant simulations in both the deterministic and statistical channels. Finally, Section~\ref{sec:Conclusion} concludes the paper.

\section{System Model}
\label{sec:System_Model}
We take into account the 5G NR system, where the gNB deploys multiple antenna panels to support transmit beamforming as described in Fig. \ref{fig:System_Model}. Based on the current research trend of assumptions and 5G NR standards, the antenna panel has multiple antenna elements; each cross represents collocated dual-polarization antenna elements with fixed polarization, $\pm 45^{\rm o}$ \cite{Roh_5GBeamforming_ComMag14, Nam_FD-MIMO, Nam_FD-MIMO_Feasibility_JSAC17, Kwon_Oh_VTC20Fall}.
Both the line-of-sight (LoS) and non-line-of-sight (NLoS) components are considered as effective ones on the received signal resulted in MPS-Beamforming \cite{Roh_5GBeamforming_ComMag14, Nam_FD-MIMO, Nam_FD-MIMO_Feasibility_JSAC17}.

The gNB in Fig. \ref{fig:System_Model} utilizes uniform linear phase antenna subarrays, i.e., antenna elements in a column with single polarization among $\pm 45^{\rm o}$. Further, the gNB can transmit a single data stream utilizing dual Tx beams, which have $+45^{\rm o}$ and $-45^{\rm o}$ polarization at the side of the gNB. On the other hand, the receive antenna polarization at the UE can be changed owing to the movement of the user. Without the loss of generality, the polarization vector of the received signal supported by $-45^{\rm o}$ polarization beamforming at the gNB is set to be aligned with a vector $\overline{a_x}$, and $\overline{a_y}$ is for the received signal's polarization caused by $+45^{\rm o}$ polarization beamforming at the gNB. The polarization of the received signals generated by $\pm 45^{\rm o}$ polarization beamforming can be regarded as orthogonal in particular, for the LoS path \cite{Kwon_Stuber11_TVT}. For the NLoS path, two transmit beams with orthogonal polarization have changed polarization after being reflected; however, the theoretical research on polarized channels in \cite{Kwon_Stuber11_TVT} reports that changed polarization maintains the orthogonality, based on the detailed mathematical analysis in the geometry based channel model along with good agreement with other empirical report. \cite{Kwon_Stuber11_TVT,Landmann07}

\begin{figure}[!t]
	\centering
\includegraphics[width=0.35\textwidth]{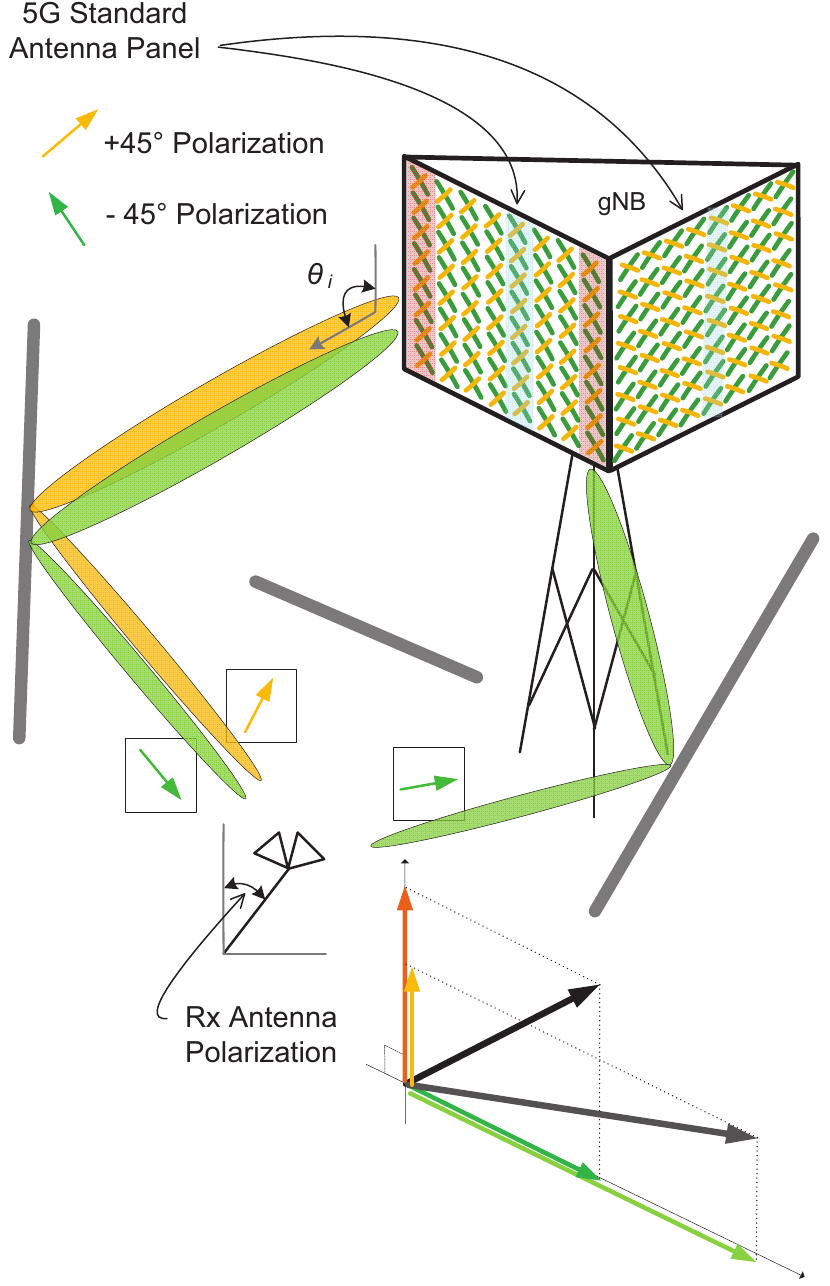}
	\caption{5G beamforming supported by multiple antenna panels at the gNB side agreed by the standard society.}
	\label{fig:System_Model}
\end{figure}

The received signals at the UE, supported by $-45^{\rm o}$ and $+45^{\rm o}$ polarization transmit beamforming, $r^{-45^{\rm o}}(t, \tau_1)$ and $r^{+45^{\rm o}}(t, \tau_2)$, respectively, can be expressed as \cite{book_Antenna_Milligan}
\setlength{\arraycolsep}{0.14em}
\begin{eqnarray}
 \label{eq:V_pol}
 r^{-45^{\rm o}}(t, \tau_1) &=& \overline{a_x} {~} E_x(t,\tau_1) {~} A(\theta_1) \nonumber\\
   &=& \overline{a_x} {~} E^{-45^{\rm o}} A(\theta_1) {~} \cos(2{\pi}f_c (t-\tau_1) + {\phi}_1) {~~~}\\ 
 r^{+45^{\rm o}}(t, \tau_2) &=& \overline{a_y} {~} E_y(t,\tau_2) {~} A(\theta_2) \nonumber\\
   &=& \overline{a_y} {~} E^{+45^{\rm o}} A(\theta_2) {~} \cos(2{\pi}f_c (t-\tau_2) + {\phi}_2), {~~~}
 \label{eq:H_pol}  
\end{eqnarray}
where for $i \in \{ 1, 2 \}$,
\begin{eqnarray}
 A(\theta_i) &=& \frac{\sin( N \psi_i /2 )}{\sin( \psi_i /2 )}{~}, {~~}\psi_i = \frac{2 \pi}{\lambda} d \cos \theta_i + \zeta_i {~}.
 \label{eq:Array_Factor}
\end{eqnarray}
\setlength{\arraycolsep}{5pt}In (\ref{eq:V_pol}) -- (\ref{eq:Array_Factor}), $E_x(t,\tau_1)$ and $E_y(t,\tau_2)$ represent polarization components of $r^{-45^{\rm o}}(t, \tau_1)$ in (\ref{eq:V_pol}) and $r^{+45^{\rm o}}(t, \tau_2)$ in (\ref{eq:H_pol}), respectively \cite{book_Antenna_Milligan}. Further, $\tau_i$, $\phi_i$, and $\theta_i$ for $i \in \{ 1, 2 \}$ are the path delay, random phase component, and the angle between the line of the linear phase array and the direction of radio propagation, respectively. Further, $f_c$ is a carrier frequency; $\zeta_i$ is the phase shift between progressive elements in the phase antenna array; and $A(\theta_i)$ is called antenna array pattern or array factor \cite{book_Antenna_Milligan}. In (\ref{eq:Array_Factor}), $\lambda$ and $d$ are the wavelength and the distance between the consecutive elements in the phase antenna array. Lastly, $E^{-45^{\rm o}}$ and $E^{+45^{\rm o}}$ are the amplitudes of the received signals excluding the antenna array factors in the directions of $\overline{a_x}$ and $\overline{a_y}${~}. Hence,
\setlength{\arraycolsep}{0.14em}
\begin{eqnarray}
  \label{eq:Final_Amplitude}
  E_x(t,\tau_1) &=& E^{-45^{\rm o}} \cos(2{\pi}f_c (t-\tau_1) + {\phi}_1),\\
  E_y(t,\tau_2) &=& E^{+45^{\rm o}} \cos(2{\pi}f_c (t-\tau_2) + {\phi}_2).
\end{eqnarray}

\setlength{\arraycolsep}{5pt}The baseband signal representation is used  for both theoretical analysis and simulation in this paper, while (\ref{eq:V_pol}) -- (\ref{eq:Array_Factor}) focus on the expression of polarization for radio propagation. Aligned with 4G LTE and 5G NR standards, the OFDM system is the fundamental scenario in this paper, and the baseband signals are regarded as the result of LoS along with non-line-of-sight (NLoS) components.  After RF demodulation and discrete Fourier transformation (DFT) at the Rx, the received signals on the $n$-th OFDM subcarrier at the antennas of the Rx with $-45^{\rm o}$ and $45^{\rm o}$ polarization, $Y_{n}^{-45^{\rm o}}$ and $Y_{n}^{45^{\rm o}}$, respectively, can be obtained. In the case that no cross-polarized signal exists, $Y_{n}^{-45^{\rm o}}$ and $Y_{n}^{45^{\rm o}}$, are respectively,
\setlength{\arraycolsep}{0.14em}
\begin{eqnarray}
 \label{eq:Y_n_info_V-Pol}
 Y_{n}^{-45^{\rm o}} &=& \sqrt{E_{s}} H_{n}^{(-45^{\rm o}, -45^{\rm o})} s(n) + w_{n}^{-45^{\rm o}} ~, \\
 \label{eq:Y_n_info_H-Pol}
 Y_{n}^{45^{\rm o}} &=& \sqrt{E_{s}} H_{n}^{(45^{\rm o}, 45^{\rm o})} s(n) + w_{n}^{45^{\rm o}} ~,
\end{eqnarray}
\setlength{\arraycolsep}{5pt}where $E_{s}$ is the energy of the transmitted information symbol $s(n)$, and $n$ is the OFDM subcarrier index. $w_{n}^{-45^{\rm o}}$ and $w_{n}^{45^{\rm o}}$ are, respectively, the noise at the Rx antennas with $-45^{\rm o}$ and $45^{\rm o}$ polarization. Finally, $H_{n}^{(-45^{\rm o}, -45^{\rm o})}$ and $H_{n}^{(45^{\rm o}, 45^{\rm o})}$ are frequency-domain channel coefficients between the Tx and Rx antennas with $-45^{\rm o}$ and $45^{\rm o}$ polarization, respectively. The cross-polarized received signal components will be taken into account in detail in Section \ref{sec:MPS-Beamforming}.

\section{Energy-Efficient Multi-Polarization Superposition Beamforming (MPS-Beamforming)}
\label{sec:MPS-Beamforming}

The system model in Section \ref{sec:System_Model} implies that the polarization of the superimposed transmit beams is primarily dependent on $E^{-45^{\rm o}}$ and $E^{+45^{\rm o}}$ in (\ref{eq:V_pol}) -- (\ref{eq:H_pol}). In this section, we propose the novel scheme of determining transmit power allocation ratio to match the polarization of the superimposed received signal to the Rx antenna polarization in a statistical sense. It is also verified in Section \ref{sec:Simulation_Results} that the proposed scheme in this section shows the best performance in terms of SER or SNR gain.

\subsection{Novel Scheme of XPD/XPR-Aware Transmit Power Allocation for MPS-Beamforming}
Superimposing two transmit beams with orthogonal polarization, $\pm 45^{\rm o}$, causes different polarization from $\pm 45^{\rm o}$ at the Rx of the UE. Furthermore, tuning the transmit power allocation ratio between two Tx beams with $\pm 45^{\rm o}$ polarization must reconfigurate the polarization of the received signal, due to the cross-polarization component of the received signal. The superposition of two transmit beams with $\pm 45^{\rm o}$ polarization causes the Rx to receive the signal with a variety of polarization via changing the transmit power allocation ratio. For this reason, the proposed scheme is called multi-polarization superposition beamforming (MPS-beamforming) in this paper. It is noteworthy that the proposed superposition beamforming scheme with numerous transmit power ratios supports multiple polarization of the signal at the Rx.
 
Based on this motivation, this section begins with the mathematical representation of the received signal caused by superimposing two transmit beams with $\pm 45^{\rm o}$ polarization and the transmit power allocation ratio between those transmit beams. The OFDM system is the fundamental scenario in both the theoretical analysis and simulation in this paper to be aligned with the present 4G LTE and 5G NR standards.

Based on the aforementioned rationale, the $-45^{\rm o}$ and $45^{\rm o}$ polarization components of the received signal on the $n$-th OFDM subcarrier, $Y_{n}^{-45^{\rm o}}$ and $Y_{n}^{45^{\rm o}}$ are expressed as
\setlength{\arraycolsep}{0.14em}
\begin{eqnarray}
 \label{eq:Y_n_info_V-Pol}
 Y_{n}^{-45^{\rm o}} &=& \sqrt{\alpha E_{s}} H_{n}^{(-45^{\rm o}, -45^{\rm o})} s(n) \nonumber\\
   &&+ \sqrt{\beta E_{s}} H_{n}^{(-45^{\rm o}, 45^{\rm o})} s(n) + w_{n}^{-45^{\rm o}} ~, \\
 \label{eq:Y_n_info_H-Pol}
 Y_{n}^{45^{\rm o}} &=& \sqrt{\alpha E_{s}} H_{n}^{(45^{\rm o}, -45^{\rm o})} s(n) \nonumber\\
   &&+ \sqrt{\beta E_{s}} H_{n}^{(45^{\rm o}, 45^{\rm o})} s(n) + w_{n}^{45^{\rm o}} ~,
\end{eqnarray}
\setlength{\arraycolsep}{5pt}where the transmit power allocation ratio, $\alpha$ and $\beta$ satisfy $\beta = 1 - \alpha$.  In the scenario without the proposed MPS-beamforming, $\alpha = 1 {\rm ~or~} 0$; whereas in the scenario of MPS-Beamforming, $0 \leq \alpha \leq 1$; consequently, $0 \leq \beta = 1 - \alpha \leq 1$.
The XPD of the received signal is the power ratio between the copolarization component to cross-polarization component of the received signal, and the statistical XPD of the received signal resulted by MPS-beamforming, $\overline{\rm XPD}^{\rm MPS}$ is defined as
\setlength{\arraycolsep}{0.14em}
\begin{eqnarray}
  \overline{\rm XPD}^{\rm MPS} &\triangleq&
     \frac{ {\rm E} \big[ | \sqrt{\alpha E_{s}} H_{n}^{(-45^{\rm o}, -45^{\rm o})} +
        \sqrt{\beta E_{s}} H_{n}^{(-45^{\rm o}, 45^{\rm o})} |^2 \big] }
        { {\rm E} \big[ | \sqrt{\alpha E_{s}} H_{n}^{(45^{\rm o}, -45^{\rm o})} +
        \sqrt{\beta E_{s}} H_{n}^{(45^{\rm o}, 45^{\rm o})} |^2 \big] }, {~~~~~}
     \label{eq:XPD_MPS_Definition}
\end{eqnarray}
\setlength{\arraycolsep}{5pt}where the operation $\rm E [\cdot]$ is the expectation or equivalently, mean over whole subcarriers.

For further analysis of (\ref{eq:XPD_MPS_Definition}), the statistical Rx cross-polarization discrimination (XPD) without MPS-Beamforming over all subcarriers for the Tx antenna polarization, $-45^{\rm o}$ and $45^{\rm o}$ is defined as
\setlength{\arraycolsep}{0.14em}
\begin{eqnarray}
 \label{eq:XPD_Stat_N}
 &&\overline{\rm XPD}^{\rm N} = \overline{\rm XPD}^{(- 45^{\rm o})} 
   \triangleq \frac{ {\rm E}[ |H_{n}^{(-45^{\rm o}, -45^{\rm o})}|^2 ] }{ ~{\rm E}[ |H_{n}^{(45^{\rm o}, {~}-45^{\rm o})}|^2 ] ~} \\
 \label{eq:XPD_Stat_N_2}
 &&{~~~~~~~~~~~~~~~~~~~~~~~} =\frac{ \sum \limits_{n=1}^{N} \big| H_{n}^{(-45^{\rm o}, - 45^{\rm o})} \big|^2 }{ \sum \limits_{n=1}^{N} \big| H_{n}^{(45^{\rm o}, - 45^{\rm o})} \big|^2 }, \\
 &&\overline{\rm XPD}^{\rm P} = \overline{\rm XPD}^{(45^{\rm o})}
 \label{eq:XPD_Stat_P}   
   {~}\triangleq \frac{ {\rm E}[ |H_{n}^{(-45^{\rm o}, 45^{\rm o})}|^2 ] }{ ~{\rm E}[ |H_{n}^{(45^{\rm o}, {~}45^{\rm o})}|^2 ] ~} \\
 \label{eq:XPD_Stat_P_2}    
 &&{~~~~~~~~~~~~~~~~~~~~~~} = \frac{ \sum \limits_{n=1}^{N} \big| H_{n}^{(-45^{\rm o}, 45^{\rm o})} \big|^2 }{ \sum \limits_{n=1}^{N} \big| H_{n}^{(45^{\rm o}, 45^{\rm o})} \big|^2 },   
\end{eqnarray}
\setlength{\arraycolsep}{5pt}where the superscript N or P denotes $-45^{\rm o}$ or $+45^{\rm o}$ polarization at the Tx antenna; not at the Rx antenna. As described in (\ref{eq:XPD_Stat_N}) and (\ref{eq:XPD_Stat_P}), we define the statistical Rx XPD as the ratio of mean channel gain at the $-45^{\rm o}$ polarized Rx antenna to the mean channel gain at the $45^{\rm o}$ polarized Rx antenna for a given Tx antenna whether it has $-45^{\rm o}$ or $+45^{\rm o}$ polarization. Since the total number of subcarriers in the overall OFDM system is the same regardless of the Rx antenna polarization, we can reach (\ref{eq:XPD_Stat_N_2}) and (\ref{eq:XPD_Stat_P_2}) from (\ref{eq:XPD_Stat_N}) and (\ref{eq:XPD_Stat_P}), respectively.

The received signal power ratio at the Rx antenna with $-45^{\rm o}$ polarization for the signals from Tx antennas with $-45^{\rm o}$ and $45^{\rm o}$ polarization is defined as Rx cross-polarization ratio (XPR). In the similar manner with the statistical Rx XPD, we define the statistical Rx XPR as
\setlength{\arraycolsep}{0.14em}
\begin{eqnarray}
 \label{eq:XPR_Stat_N}
  &&\overline{\rm XPR}^{\rm N} = \overline{\rm XPR}^{(-45^{\rm o})}
    \triangleq \frac{ {\rm E}[ |H_{n}^{(-45^{\rm o}, -45^{\rm o})}|^2 ] }{ ~{\rm E}[ |H_{n}^{(-45^{\rm o}, {~}45^{\rm o})}|^2 ] ~} {~} \\
 \label{eq:XPR_Stat_N_2}
  &&{~~~~~~~~~~~~~~~~~~~~~~~} = \frac{ \sum \limits_{n=1}^{N} \big| H_{n}^{(-45^{\rm o}, - 45^{\rm o})} \big|^2 }{ \sum \limits_{n=1}^{N} \big| H_{n}^{(- 45^{\rm o}, 45^{\rm o})} \big|^2 }, \\ 
  \label{eq:XPR_Stat_P} 
  && \overline{\rm XPR}^{\rm P} = \overline{\rm XPR}^{(45^{\rm o})}
  {~}\triangleq \frac{ {\rm E}[ |H_{n}^{(45^{\rm o}, -45^{\rm o})}|^2 ] }{ ~{\rm E}[ |H_{n}^{(45^{\rm o}, {~}45^{\rm o})}|^2 ] ~} {~} \\
 \label{eq:XPR_Stat_P_2}
  &&{~~~~~~~~~~~~~~~~~~~~~~} = \frac{ \sum \limits_{n=1}^{N} \big| H_{n}^{(45^{\rm o}, - 45^{\rm o})} \big|^2 }{ \sum \limits_{n=1}^{N} \big| H_{n}^{(45^{\rm o}, 45^{\rm o})} \big|^2 } {~}.
\end{eqnarray}
\setlength{\arraycolsep}{5pt}In (\ref{eq:XPR_Stat_N}) and (\ref{eq:XPR_Stat_P}), the statistical Rx XPR can be interpreted as the ratio between the mean channel gain from the $-45^{\rm o}$ polarized Tx antenna to the given Rx antenna; and the mean channel gain from the $45^{\rm o}$ polarized Tx antenna to the same Rx antenna.
We directly use (\ref{eq:XPD_Stat_N_2}), (\ref{eq:XPD_Stat_P_2}), (\ref{eq:XPR_Stat_N_2}) and (\ref{eq:XPR_Stat_P_2}) to reach (\ref{eq:XPD_MPS_Expression}) from the definition of $\overline{\rm XPD}^{\rm MPS}$ in (\ref{eq:XPD_MPS_Definition}). 

The following Lemma describes the relation between $\overline{\rm XPD}^{\rm N}$, $\overline{\rm XPD}^{\rm P}$, $\overline{\rm XPR}^{\rm N}$ and $\overline{\rm XPR}^{\rm P}$, and is utilized to derive XPD/XPR-aware transmission power ratio, $\alpha_{\rm XPD/XPR}$ and $\beta_{\rm XPD/XPR}$ in (\ref{eq:Optimal_alpha}).
\begin{lemma}
  \begin{equation}
    \frac{ ~\overline{\rm XPD}^{\rm N} }{ ~\overline{\rm XPD}^{\rm P} } = \frac{ ~\overline{\rm XPR}^{\rm N} }{ ~\overline{\rm XPR}^{\rm P} }
  \end{equation}
\end{lemma}
\begin{proof}
\setlength{\arraycolsep}{0.14em}
  \begin{eqnarray}
  \label{eq:Lemma_XPD_XPR}
    \frac{ ~\overline{\rm XPD}^{\rm N} }{ ~\overline{\rm XPD}^{\rm P} }
     {}&=&{} \frac{ ~ {\rm E}[ |H_{n}^{(-45^{\rm o}, -45^{\rm o})}|^2 ] ~ \big/ ~ {\rm E}[ |H_{n}^{(45^{\rm o}, -45^{\rm o})}|^2 ] ~ }{ ~ {\rm E}[ |H_{n}^{(-45^{\rm o}, 45^{\rm o})}|^2 ] ~ \big/ ~ {\rm E}[ |H_{n}^{(45^{\rm o}, 45^{\rm o})}|^2 ] ~ } \\
  \label{eq:Lemma_XPD_XPR_2}    
    {}&=&{} \frac{ ~ {\rm E}[ |H_{n}^{(-45^{\rm o}, -45^{\rm o})}|^2 ] ~ \big/ ~ {\rm E}[ |H_{n}^{(-45^{\rm o}, 45^{\rm o})}|^2 ] ~ }{ ~ {\rm E}[ |H_{n}^{(45^{\rm o}, -45^{\rm o})}|^2 ] ~ \big/ ~ {\rm E}[ |H_{n}^{(45^{\rm o}, 45^{\rm o})}|^2 ] ~ } \\
  \label{eq:Lemma_XPD_XPR_3}  
    {}&=&{} \frac{ ~\overline{\rm XPR}^{\rm N} }{ ~\overline{\rm XPR}^{\rm P} }
  \end{eqnarray}
\end{proof}
\setlength{\arraycolsep}{5pt}\noindent The numerator and denominator in (\ref{eq:Lemma_XPD_XPR}) are from the definition of the statistical Tx XPD, $\overline{\rm XPD}^{\rm N}$ and $\overline{\rm XPD}^{\rm P}$ in (\ref{eq:XPD_Stat_N}) and (\ref{eq:XPD_Stat_P}), respectively; four expectation components in (\ref{eq:Lemma_XPD_XPR}) are rearranged as described in (\ref{eq:Lemma_XPD_XPR_2}). Finally, the numerator and denominator of (\ref{eq:Lemma_XPD_XPR_2}) are aligned with the definitions of $\overline{\rm XPR}^{\rm N}$ and $\overline{\rm XPR}^{\rm P}$ in (\ref{eq:XPR_Stat_N}) and (\ref{eq:XPR_Stat_P}), respectively. The ratio of statistical Rx XPD, $\big( \overline{\rm XPD}^{\rm N} / {~}\overline{\rm XPD}^{\rm P} \big)$ is the same as the ratio of statistical Rx XPR, $\big( \overline{\rm XPR}^{\rm N} /{~}\overline{\rm XPR}^{\rm P} \big)$ as presented in (\ref{eq:Lemma_XPD_XPR_3}).

Utilizing (\ref{eq:XPD_Stat_N}) -- (\ref{eq:Lemma_XPD_XPR_3}), $\overline{\rm XPD}^{\rm MPS}$ in (\ref{eq:XPD_MPS_Definition}) is derived as
\setlength{\arraycolsep}{0.14em}
\begin{eqnarray}
  &&\overline{\rm XPD}^{\rm MPS} \nonumber \\
  && \triangleq
     \frac{ {\rm E} \big[ | \sqrt{\alpha E_{s}} H_{n}^{(-45^{\rm o}, -45^{\rm o})} +
        \sqrt{\beta E_{s}} H_{n}^{(-45^{\rm o}, 45^{\rm o})} |^2 \big] }
        { {\rm E} \big[ | \sqrt{\alpha E_{s}} H_{n}^{(45^{\rm o}, -45^{\rm o})} +
        \sqrt{\beta E_{s}} H_{n}^{(45^{\rm o}, 45^{\rm o})} |^2 \big] }
     \label{eq:XPD_MPS_Definition_Repeat} \nonumber\\
  && = \frac{ \alpha + \beta \big/ {~}\overline{\rm XPR}^{\rm N} }
    { \alpha \big/ {~}\overline{\rm XPD}^{\rm N} + \beta \big/\big({~}\overline{\rm XPD}^{\rm P} {~} \overline{\rm XPR}^{\rm N} \big) } \label{eq:XPD_MPS_Expression} \\
  && = \frac { \alpha \big( 1 - 1 \big/ {~}\overline{\rm XPR}^{\rm N} \big)  + 1 \big/ {~}\overline{\rm XPR}^{\rm N}  }
    { \alpha \big( 1 \big/ \overline{\rm XPD}^{\rm N} - 1 \big/ \big(\overline{\rm XPD}^{\rm P} \overline{\rm XPR}^{\rm N} \big)  \big) + 1 \big/\big(\overline{\rm XPD}^{\rm P} \overline{\rm XPR}^{\rm N} \big) } {~}, \label{eq:XPD_MPS_Fn_of_alpha} \nonumber \\
  {~}  
\end{eqnarray}
\setlength{\arraycolsep}{5pt}where once again, $n$ is the subcarrier index. Substituting $1-\alpha$ for $\beta$, we can derive (\ref{eq:XPD_MPS_Fn_of_alpha}) from (\ref{eq:XPD_MPS_Expression}).  It is noteworthy that $\overline{\rm XPD}^{\rm MPS}$ is the function of $\alpha$, monotonically increasing or decreasing as shown in (\ref{eq:XPD_MPS_Fn_of_alpha}). 

Finally, the Tx allocates the transmission power to the Tx antennas with $-45^{\rm o}$ and $45^{\rm o}$ polarization in the conventional scenario of the fixed total transmission power constraint.
In other words, the Tx determines transmit power allocation ratio, $\alpha$ and consequently, $\beta = 1 - \alpha$; and the objective is to align $\overline{\rm XPD}^{\rm MPS}$ with the Rx antenna polarization, $\overline{\rm XPD}^{\rm Rx-Ant}$, i.e.,
\begin{equation}
  \overline{\rm XPD}^{\rm MPS} = \overline{\rm XPD}^{\rm Rx-Ant} {~}.
  \label{eq:Pol_matching}
\end{equation}
The Rx antenna polarization $\overline{\rm XPD}^{\rm Rx-Ant}$ can be changed by the rotation of antenna origination at the UE, caused by the movement of the UE. Nonetheless, it is worth of notice that the polarization state information (PSI) estimation such as XPD and XPR at the Rx can still be accomplished based on the rotated antenna origination. In this scenario, $\overline{\rm XPD}^{\rm Rx-Ant} = \infty$ or $\overline{\rm XPD}^{\rm Rx-Ant} = 0$ corresponding to the Rx antenna polarization angle $-45^{\rm o}$ or $45^{\rm o}$, respectively. In other words, the Rx does not need to maintain $-45^{\rm o}$ and $45^{\rm o}$ Rx antenna polarization angles in a physical sense; whereas, even the rotated antenna origination can be regarded as the basis of PSI estimation such as XPD and XPR. The change of the Rx antenna polarization angle from the basis of PSI estimation can be also measured by the gyroscope embedded in the current commercial UEs such as smartphones. The estimated PSI is reported to gNB through the feedback channel allocated for sharing the PSI in the similar manner with the present conventional CSI feedback.

Based on (\ref{eq:XPD_MPS_Fn_of_alpha}) and (\ref{eq:Pol_matching}), the theoretical XPD/XPR-aware transmission power allocation ratio, $\alpha_{\rm XPD/XPR}$ and $\beta_{\rm XPD/XPR}$ can be derived to have the closed form expression as following.
\setlength{\arraycolsep}{0.14em}
\begin{eqnarray}
 \label{eq:Optimal_alpha}
 \hat{\alpha}_{\rm XPD/XPR} {}&=&{} \dfrac{  \overline{{\rm XPD}}^{\rm Rx-Ant} - \overline{{\rm XPD}}^{\rm P}  }{  \overline{{\rm XPD}}^{\rm P}(\overline{{\rm XPR}}^{\rm N}-1)+\overline{{\rm XPD}}^{\rm Rx-Ant}(1-\overline{{\rm XPR}}^{\rm P})  }{~}, \nonumber \\
 ~\\
 \alpha_{\rm XPD/XPR} {}&=&{} \min \big\{1, {~}\max \{0, {~}\hat{\alpha}_{\rm XPD/XPR} \} \big\}{~}, \label{eq:Suboptimal_alpha}\\
 \beta_{\rm XPD/XPR} {}&=&{} 1- \alpha_{\rm XPD/XPR} {~} \label{eq:Optimal_beta}.
\end{eqnarray}
\setlength{\arraycolsep}{5pt}In the scenario that the theoretical XPD/XPR-aware transmit power ratio $\hat{\alpha}_{\rm XPD/XPR}$ is not in the range of $[0,1]$, the closest value to boundary values, zero or unity, is selected as (\ref{eq:Suboptimal_alpha}), for the reason that $\overline{\rm XPD}^{\rm MPS}$ is the function of $\alpha$, and monotonically increases or decreases with respect to $\alpha$ as shown in (\ref{eq:XPD_MPS_Fn_of_alpha}). 

\subsection{Rotation and Eccentricity of the Polarization Ellipse in MPS-Beamforming}
\label{sec:Pol_ellipse}

The mathematical derivation in this section describes that the polarization of MPS-Beamforming can be the elliptical polarization as portrayed in Fig. \ref{fig:Polarization_Ellipse} even for the line-of-sight (LoS) scenario. 

At the side of the UE, the superposition of the two signals, $r^{-45^{\rm o}}(t, \tau_1)$ and $r^{+45^{\rm o}}(t, \tau_2)$ in (\ref{eq:V_pol}) and (\ref{eq:H_pol}), respectively, impinges on the Rx antenna. That is,
\setlength{\arraycolsep}{0.14em}
\begin{eqnarray}
 \label{eq:r}
 r(t, \tau_1, \tau_2) &=& r^{-45^{\rm o}}(t, \tau_1) + r^{+45^{\rm o}}(t, \tau_2) \\
 &=&\overline{a_x}E^{-45^{\rm o}} A(\theta_1) \cos({\varphi})+\overline{a_y}E^{+45^{\rm o}} A(\theta_2) \cos({\varphi}+{\Delta}) \nonumber
\end{eqnarray}
where
\begin{eqnarray}
 \label{eq:phi_variance}
 \varphi &=& 2 \pi f_c (t - \tau_1) + {\phi}_1 {~}, \\
 \label{eq:Delta}
 \Delta &=& - 2 \pi f_c (\tau_2 - \tau_1) + ({\phi}_2-{\phi}_1) {~}.  
\end{eqnarray}
\setlength{\arraycolsep}{5pt}The phase difference between two received signals supported by $\pm 45^{\rm o}$ polarization beamforming is $\Delta$; even in the scenario of the line-of-sight (LoS) link between the Tx and Rx, the phase difference may not be zero due to the mismatching between $\pm 45^{\rm o}$ polarization Tx antenna elements.

It is noteworthy that gNB can rotate the polarization ellipse, i.e., the major and minor axes of the polarization ellipse at the UE side via utilizing the proposed scheme, MPS-Beamforming. Furthermore, MPS-Beamforming can also change the eccentricity of the polarization ellipse; as the polarization ellipse becomes narrow and long, i.e., eccentricity comes to be very high, more received signal power can be concentrated on the direction of polarization ellipse's major axis. 
\begin{figure}[ht]
	\centering
\includegraphics[width=0.41\textwidth, height=1.6in]{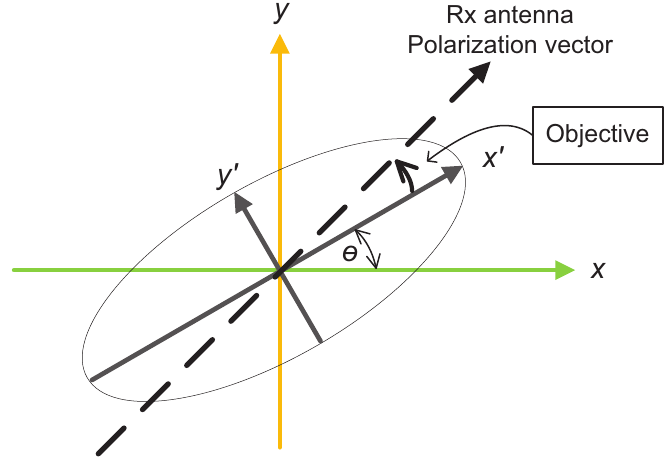}
	\caption{Polarization ellipse and the rotation of the ellipse.}
	\label{fig:Polarization_Ellipse}
\end{figure}

Based on (\ref{eq:V_pol}) -- (\ref{eq:H_pol}) and (\ref{eq:r}) with the auxiliary variables, $\varphi$ and $\Delta$ in (\ref{eq:phi_variance}) -- (\ref{eq:Delta}), utilizing the trigonometric identity for $\cos(\varphi + \Delta)$ in (\ref{eq:r}), the polarization of $r(t, \tau_1, \tau_2)$, the superimposed signal at the UE, satisfies that
\setlength{\arraycolsep}{0.14em}
\begin{eqnarray}
 \label{eq:H_pol_mod}
 &&\dfrac{E_y(t, \tau_2)}{E^{+45^{\rm o}}} = \cos{\varphi}\cos{\Delta}-\sin{\varphi}\sin{\Delta}, \\
 \label{eq:V_pol_mod} 
 &&\dfrac{E_x(t, \tau_1)}{E^{-45^{\rm o}}} = \cos{\varphi}.
\end{eqnarray}
\setlength{\arraycolsep}{5pt}Further, (\ref{eq:H_pol_mod}) and (\ref{eq:V_pol_mod}) can be modified as
\setlength{\arraycolsep}{0.14em}
\begin{eqnarray}
 \label{eq:H_V_pol_relation}
 && \bigg(\dfrac{E_x(t,\tau_1)}{E^{-45^{\rm o}}\sin{\Delta}} - \dfrac{\cos \varphi \cos \Delta}{\sin \Delta} \bigg)^{2} = \sin^2 \varphi \\
 \label{eq:sin^2_phi}
 &&{~~~} \cos{\varphi} = \dfrac{E_x(t, \tau_1)}{E^{-45^{\rm o}}}~,~~ \sin^2 {\varphi} = 1 - \bigg(\dfrac{E_x(t, \tau_1)}{E^{-45^{\rm o}}} \bigg)^2. 
\end{eqnarray}
\setlength{\arraycolsep}{5pt}

We plug in (\ref{eq:sin^2_phi}) to (\ref{eq:H_V_pol_relation}); then,
\setlength{\arraycolsep}{0.14em}
\begin{eqnarray} 
 \label{eq:Pre-Ellipse}  
 \bigg(\dfrac{E_x(t,\tau_1)}{E^{-45^{\rm o}}\sin{\Delta}}\bigg)^{2} &+& \bigg(\dfrac{E_y(t,\tau_2)}{E^{+45^{\rm o}}\sin{\Delta}}\bigg)^{2} \nonumber \\
 &-& 2\dfrac{E_x(t,\tau_1)E_y(t,\tau_2)}{E^{-45^{\rm o}} E^{+45^{\rm o}}}\dfrac{\cos{\Delta}}{\sin^{\rm2}{\Delta}}=1{~}, 
\end{eqnarray}
\setlength{\arraycolsep}{5pt}which is in concord with the expression of the rotated ellipse.
Finally, rotating the coordinates verifies that the polarization of the superimposed signal at the UE is elliptic as
\begin{equation}
 \label{eq:Ellipse}
 \bigg(\dfrac{E_x'(t, \tau_1)}{a}\bigg)^2 + \bigg(\dfrac{E_y'(t, \tau_2)}{b}\bigg)^2 = 1,
\end{equation}where
\setlength{\arraycolsep}{0.14em}\begin{eqnarray}
 \label{eq:Rotation matrix}
 \bigg[\begin{matrix}
  E_x'(t, \tau_1) \\  E_y'(t, \tau_2)
 \end{matrix}\bigg]  &=&
 \bigg[\begin{matrix} 
 {~~} \cos{\theta} & {~}\sin{\theta} {~} \\ -\sin{\theta} & {~}\cos{\theta} {~}
 \end{matrix}\bigg]   
 \bigg[\begin{matrix}
  E_x (t, \tau_1) \\ E_y (t, \tau_2) 
 \end{matrix}\bigg].
\end{eqnarray}\setlength{\arraycolsep}{5pt}

In the comparison of (\ref{eq:Pre-Ellipse}) with (\ref{eq:Ellipse}) after plugging in (\ref{eq:Rotation matrix}) into (\ref{eq:Ellipse}), the rotation angle of the polarization ellipse, $\theta$ and the squared eccentricity of the polarization ellipse, $\epsilon ^2$ are described as
\setlength{\arraycolsep}{0.14em}
\begin{eqnarray}
 \label{eq:Theta}
 {\theta} {}&=&{} \dfrac{1}{2}\tan^{\rm-1}\bigg(\dfrac{2\cos{\Delta}}{\tfrac{E^{-45^{\rm o}}}{E^{+45^{\rm o}}}-\tfrac{E^{+45^{\rm o}}}{E^{-45^{\rm o}}}}\bigg) \\
 \label{eq:Theta_2}
 &=& \dfrac{1}{2}\tan^{-1}\bigg(\dfrac{ 2\cos{\Delta} }{ \sqrt{\rm XPD} - 1/\sqrt{\rm XPD}}\bigg) {~}, \\ 
 \label{eq:Eccentricity}
 \epsilon ^2 &\triangleq & 1 - \dfrac{b^2}{a^2} \\
 \label{eq:Eccentricity_2}
   &=& 1 - \dfrac{\tfrac{E^{+45^{\rm o}}}{E^{-45^{\rm o}}}(1+\sec2{\theta})+\tfrac{E^{-45^{\rm o}}}{E^{+45^{\rm o}}}(1-\sec2{\theta})}{\tfrac{E^{+45^{\rm o}}}{E^{-45^{\rm o}}}(1-\sec2{\theta})+\tfrac{E^{-45^{\rm o}}}{E^{+45^{\rm o}}}(1+\sec2{\theta})} \\
 \label{eq:Eccentricity_3}
   &=& 1 - \frac{ 1 + \sec2{\theta} + {\rm XPD}(1-\sec2{\theta}) }{ 1 - \sec2{\theta} + {\rm XPD}(1+\sec2{\theta}) } \nonumber\\
 \label{eq:Eccentricity_4}  
   &=& \frac{ 2 ({\rm XPD} - 1) \sec2{\theta} }{ 1 - \sec2{\theta} + {\rm XPD}(1+\sec2{\theta}) } {~}.
\end{eqnarray}
\setlength{\arraycolsep}{5pt}We define the instantaneous Rx XPD as
\begin{equation}
 \label{eq:Rx_XPD}
 {\rm XPD} = \bigg( \dfrac{ E^{-45^{\rm o}} }{ E^{+45^{\rm o}} }\bigg)^2 ,
\end{equation}
and plug in the instantaneous Rx XPD to (\ref{eq:Theta}) and (\ref{eq:Eccentricity_2}) so that we can reach (\ref{eq:Theta_2}) and (\ref{eq:Eccentricity_4}). The XPD is the power ratio of two orthogonal polarization components, which are $E^{-45^{\rm o}}$ and $E^{+45^{\rm o}}$ in this section.
The fundamental definition of the eccentricity in the ellipse is expressed with the squared format in (\ref{eq:Eccentricity}). It is worth mentioning that both the rotation angle and the squared eccentricity of the polarization ellipse, $\theta$ and $\epsilon^2$, respectively, are the functions of $( E^{-45^{\rm o}}/E^{+45^{\rm o}} )$; therefore, the instantaneous Rx XPD.  The eccentricity, denoted by $\epsilon$, is the metric to estimate the deviation of the conic section from the circle.  The squared eccentricity, $\epsilon ^2$ is utilized in this paper for the simplicity of the expression rather than using the eccentricity, $\epsilon$; it varies from zero to unity. As $\epsilon ^2$ approaches zero, the ellipse converges to a circle; while as $\epsilon ^2$ approaches unity, the ellipse converges to a linear line.

Furthermore, the direct relation between the squared eccentricity, $\epsilon ^2$ and the rotation angle of the polarization ellipse, $\theta$ is derived as following via utilizing (\ref{eq:Theta_2}) and (\ref{eq:Eccentricity_4}).
\setlength{\arraycolsep}{0.14em}
\begin{eqnarray}
 \label{eq:Eccen(Theta)}
 {\epsilon ^2} &=& 1 - \dfrac{\sqrt {\cos^2{\Delta}+\tan^2{2\theta}} - \cos{\Delta}\sec{2\theta}}{\sqrt {\cos^2{\Delta}+\tan^2{2\theta}}+\cos{\Delta}\sec{2\theta}} \nonumber\\
 &=& \dfrac{ 2\cos{\Delta}\sec{2\theta}}{\sqrt {\cos^2{\Delta}+\tan^2{2\theta}}+\cos{\Delta}\sec{2\theta}} {~}. 
\end{eqnarray}
\setlength{\arraycolsep}{5pt}The relation between $\epsilon^{\rm2}$ and $\theta$ in (\ref{eq:Eccen(Theta)}) for a variety of $\Delta$ defined in (\ref{eq:Delta}) is illustrated in Section \ref{sec:Simulation_Results}. 

\subsection{Fine Adjustment of the XPD in MPS-Beamforming}
This section describes the manner of fine tuning for the XPD in MPS-Beamforming to cope with practical situations such as imperfect reconfiguration of antenna polarization and channel coefficients of subcarriers varying faster than the update of subcarrier assignment. The fine tuning based on the feedback loop can compensate for the imperfect realization of XPD in MPS-Beamforming.

The theoretical XPD/XPR-aware transmit power ratio, $\alpha_{\rm XPD/XPR}$, is determined by the statistical XPD and XPR over the OFDM subcarriers as described in (\ref{eq:Optimal_alpha}) -- (\ref{eq:Optimal_beta}). Further, the MPS-Beamforming OFDM system fulfills subcarrier assignment to be described in Section \ref{sec:Subcarrier_Alloc}. For the assigned subcarriers, the channel impulse response per subcarrier can be changed before fulfilling subcarrier re-assignment in the practical system. 

The gNB can adjust its transmit XPD, $\alpha / \beta$, or equivalently, its transmit polarization angle, i.e., $\tan^{-1} \big( \sqrt{ \alpha / \beta } ~\big)$ via assigning  transmit power ratio, $\alpha$ and $\beta$, to $-45^{\rm o}$ and $+45^{\rm o}$ polarization beamforming, such that the gNB controls the rotation angle and eccentricity of the Rx polarization ellipse caused by the MPS-Beamforming at the UE side. The mechanism of adjusting the receive polarization ellipsis is illustrated in Fig. \ref{fig:Pol_adjustment} considering a scenario without the loss of generality. Assuming the Rx antenna polarization angle, $\tan^{-1} \big( \sqrt{ \overline{\rm XPD}^{\rm Rx-Ant} } ~\big)$, is $75^{\rm o}$, the best scenario of polarization matching is that the mean receive polarization angle in MPS-Beamforming, $\tan^{-1} \big( \sqrt{ \overline{\rm XPD}^{\rm MPS} } ~\big)$, is $68^{\rm o}$ in Fig. \ref{fig:Pol_adjustment}.  $\overline{\rm XPD}^{\rm Rx-Ant}$ and $\overline{\rm XPD}^{\rm MPS}$ can have different values in general, relying on the wireless channel environment and the rotation of the Rx antenna caused by UE's movement.

\begin{figure}[ht]
	\centering
\includegraphics[width=0.48\textwidth, height=2.0in]{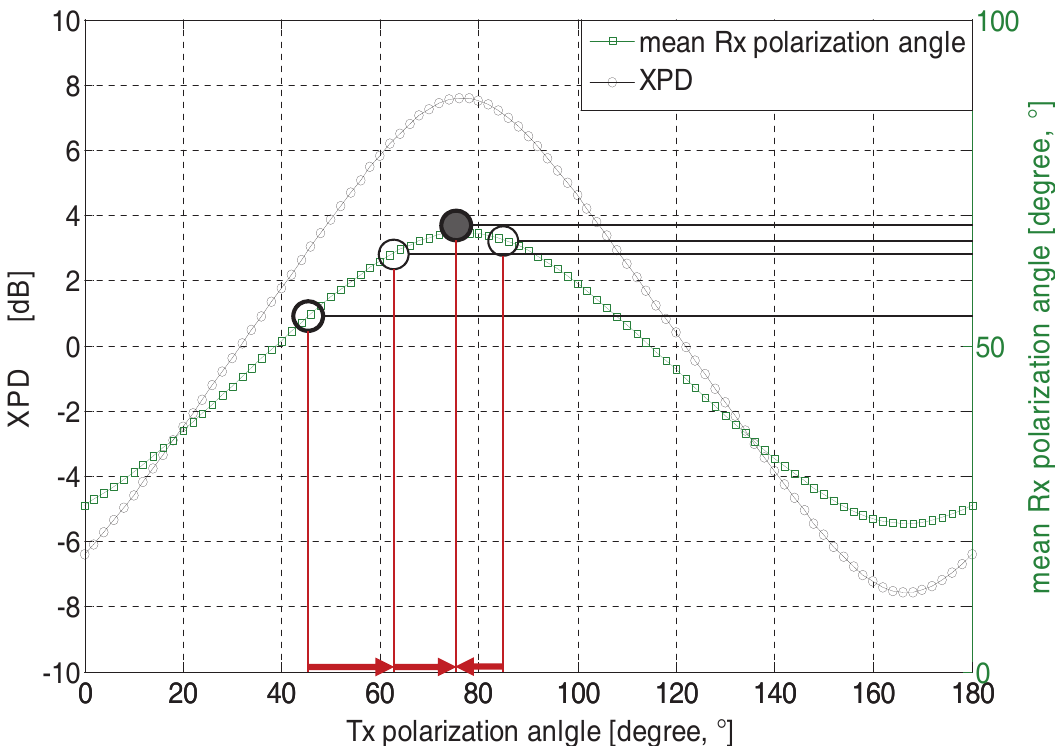}
	\caption{Adjustment of rotation angle and eccentricity of the Rx polarization ellipse.}
	\label{fig:Pol_adjustment}
\end{figure}

The UE Rx reports PSI in terms of the mean receive XPD and XPR estimated at the UE along with the current mean receive polarization angle to the gNB by the feedback channel. The estimation of the mean Rx XPD and XPR can be performed based on dual-polarized antenna elements at the Rx with the support of reference signals. In turn, the gNB changes its transmit polarization angle or transmit power ratio until the condition, $\tan^{-1} \big( \sqrt{ \overline{\rm XPD}^{\rm MPS} } ~\big)$ = $\tan^{-1} \big( \sqrt{ \overline{\rm XPD}^{\rm Rx-Ant} } ~\big)$, equivalent to (\ref{eq:Pol_matching}), or until   $\overline{\rm XPD}^{\rm MPS}$ approaches $\overline{\rm XPD}^{\rm Rx-Ant}$ as closely as possible. The UE and gNB perform iterations of PSI feedback followed by the associated fine tuning of the transmit polarization angle or equivalently, transmit power ratio. The bin of increasing/decreasing transmit polarization angle in Fig. \ref{fig:Pol_adjustment} can be dynamic depending on the difference between the reported current and previous mean receive polarization angles via the feedback channel.  The sophisticated algorithm for the adjustment of the bin of increasing/decreasing the transmit polarization angle is outside of the scope and contribution of this paper; it is our future work.


\subsection{Subcarrier Assignment for MPS-Beamforming}
\label{sec:Subcarrier_Alloc}
Subcarriers pertinent to the system/subsystem need to be selected relying on the scheme that the system/subsystem employs, and the gNB assigns subcarriers to each subsystem in its service coverage, in an efficient manner considering overall system performance and fairness. This paper proposes a new scheme of MPS-Beamforming with the consideration of the OFDM system; therefore, the corresponding algorithm of subcarrier assignment needs to be provided.
The proposed subcarrier assignment algorithm first, exclude subcarriers that have low channel gains. Then, the algorithm estimates $\overline{\rm XPD}^{\rm N}$, $\overline{\rm XPD}^{\rm P}$, $\overline{\rm XPR}^{\rm N}$ and $\overline{\rm XPR}^{\rm P}$ of the screened subcarriers; and the concomitant $\overline{\rm XPD}^{\rm MPS}$ in (\ref{eq:XPD_MPS_Definition}). Finally, the subcarriers closest to $\overline{\rm XPD}^{\rm MPS}$ are selected.


\begin{table}[!t]
	\renewcommand{\arraystretch}{1.3}
	\caption{MPS-Beamforming Subcarrier Assignment Algorithm}
	\label{Tbl:SubcarrierAlloc}
	\centering
	\begin{tabular}{|| c | p{7cm} ||}
		\hline
		\bfseries Step & \bfseries Execution\\
		\hline
		1 & Among $N_{\rm total}$ total number of subcarriers, screen out $\eta$-percentile subcarriers that have the low channel gain. In the simulation of this paper $N_{\rm total} = 2048$, and $\eta = 35$; they can be adapted to the channel condition.\\
		\hline
		2 & For the selected ($100-\eta$)-percentile subcarriers, estimate $\overline{\rm XPD}^{\rm N}$, $\overline{\rm XPD}^{\rm P}$, $\overline{\rm XPR}^{\rm N}$ and $\overline{\rm XPR}^{\rm P}$. Estimate the XPD/XPR aware transmit power allocation ratio, $\alpha_{\rm XPD/XPR}$ and $\beta_{\rm XPD/XPR}$ based on (\ref{eq:Optimal_alpha})-(\ref{eq:Optimal_beta}).\\
		\hline
		3 & Estimate $\overline{\rm XPD}^{\rm MPS}$ in (\ref{eq:XPD_MPS_Fn_of_alpha}), based on Step-2.\\
		\hline
		4 & Select $N$ number of subcarriers that have the per-subcarrier XPD in MPS-Beamforming, ${\rm XPD}_n^{\rm MPS}$ in (\ref{eq:Per-subcarrier_XPD}) closest to the estimated $\overline{\rm XPD}^{\rm MPS}$ in Step-3. $N$ is set to be 48 in the simulation of this paper, which is in agreement with the fundamental OFDM physical resource block (PRB) size in the current standards.\\
		\hline
	\end{tabular}
\end{table}

In the practical environment of the multipath fading or frequency selective fading channel, each subcarrier has different characteristics in terms of the channel gain and thus, Rx XPD and XPR. The subcarrier-dependent characteristics of the XPD/XPR should be coped with in the proposed MPS-Beamforming and transmit power allocation. Table \ref{Tbl:SubcarrierAlloc} provides novel heuristic algorithm of XPD/XPR-aware OFDM subcarrier assignment to efficiently support the proposed MPS-Beamforming. It is noteworthy that the heuristic methodology has been regarded as an applaudable approach and frequently utilized in OFDM subcarrier and other resource allocation. \cite{Cheng_MU_Subcarrier_Bit_Power_Alloc}. A single user is considered in this paper; meanwhile, multi-user (MU) OFDM will be considered for further elaborated subcarrier assignment scheme in future works.

As an additional metric based on which the gNB assigns subcarriers to the particular transmission, the per-subcarrier XPD, ${\rm XPD}_n^{\rm MPS}$ for the subcarrier index, $n$ is defined as
\setlength{\arraycolsep}{0.14em}
\begin{eqnarray}
  {\rm XPD}_n^{\rm MPS} {}&=&{} \frac{ ~ \big| \sqrt{\alpha E_{s}} H_{n}^{(-45^{\rm o}, -45^{\rm o})} + \sqrt{\beta E_{s}} H_{n}^{(-45^{\rm o}, 45^{\rm o})} \big|^2  }
        {  \big| \sqrt{\alpha E_{s}} H_{n}^{(45^{\rm o}, -45^{\rm o})} +
        \sqrt{\beta E_{s}} H_{n}^{(45^{\rm o}, 45^{\rm o})} \big|^2  } \nonumber \\
  {}&=&{} \frac{ ~ \big| \sqrt{\alpha} H_{n}^{(-45^{\rm o}, -45^{\rm o})} + \sqrt{\beta} H_{n}^{(-45^{\rm o}, 45^{\rm o})} \big|^2  }
        {  \big| \sqrt{\alpha} H_{n}^{(45^{\rm o}, -45^{\rm o})} +
        \sqrt{\beta} H_{n}^{(45^{\rm o}, 45^{\rm o})} \big|^2  } {~}.  \label{eq:Per-subcarrier_XPD}    
\end{eqnarray}

\setlength{\arraycolsep}{5pt}The subcarrier assignment algorithm in Table \ref{Tbl:SubcarrierAlloc} is applied in the simulation in Section \ref{sec:Simulation_Results}, and contributes to the improvement of the system performance in terms of SER; equivalently, in terms of energy efficiency or SNR gain for the target SER. The proposed MPS-Beamforming with the novel XPD/XPR-aware transmit power allocation ratio and XPD/XPR-aware subcarrier assignment algorithm yields the best performance in terms of the SER among any combination of two Tx beams with $\pm 45^{\rm o}$ polarization including the scenario of single Tx beam transmission with either $-45^{\rm o}$ or $45^{\rm o}$ polarization corresponding to $\alpha_{\rm XPD/XPR} = 1$ or $\alpha_{\rm XPD/XPR} = 0$, respectively. Further, the simulation results and analysis for the complexity of the proposed subcarrier assignment algorithm is provided with Fig. \ref{fig:Complexity_SubChAlloc} in Section \ref{sec:Simulation_Results}.

\section{Simulation Results}
\label{sec:Simulation_Results}
This section provides comprehensive simulation results and the associated analyses for the novel scheme of MPS-Beamforming with the XPD/XPR-aware transmit power allocation scheme and the subcarrier assignment algorithm.  Theoretical analyses in Section \ref{sec:MPS-Beamforming} are verified and strongly supported by the simulation results, via showing that adopting proposed schemes yields the best performance with the significant improvement in SER among possible transmit power allocation ratios. The abundant comparisons of the proposed scheme with the conventional scheme are accomplished via setting up the transmit power allocation ratio, $\alpha$. The conventional scheme of single beam transmission with either $-45^{\rm o}$ or $45^{\rm o}$ is included in the simulation results via realizing the scenario with $\alpha = 1$ or $\alpha = 0$, respectively. Furthermore, MPS-Beamforming with random transmit power allocation ratio is also considered to verify that the theoretically obtained XPD/XPR-aware transmit power allocation ratio shows the best performance in terms of SER. This paper is the first to propose superposition of two transmit beams with orthogonal polarization; therefore, the further comparison with other scheme of superimposing multi-polarization Tx beams is unfeasible.

Both the statistical and deterministic polarized wireless channels are utilized for the simulations based on \cite{Kwon_Stuber14_PDMA_TWC} and \cite{Kwon_Stuber11_TVT}, where the polarized channel model has been verified with the remarkable agreement with the previously reported empirical data from several channel sounding campaigns. The considered scenario is the urban or metropolitan area with the street canyon where the received signal's angles of arrival at the Rx follows verified distribution with the directivity, and beamforming and its gain are regarded as included in the channel realizations and the associated simulations based on \cite{Kwon_Stuber14_PDMA_TWC} and \cite{Kwon_Stuber11_TVT}.
Each SER curve in the statistical polarized wireless channels is the result of Monte Carlo simulation over 100 individual wireless channel realizations which have the same PSI in terms of XPD and XPR. In turn, for one individual channel realization, the simulation accomplishes $10^4$ iterations of estimating SER. This is the first paper to propose the novel MPS-Beamforming scheme; therefore, the simulations focus the primary consideration on the lowest modulation order, i.e., quadrature phase shift keying (QPSK) with the modulation order 2, in the current 5G NR standards. 

The significant improvement of the SER performance in the OFDM system adopting MPS-Beamforming with XPD/XPR-aware transmit power allocation, is illustrated in Figs. \ref{fig:Statistic_SER_0_RxAnt} and \ref{fig:Statistic_SER_0_RxAnt_WO_SubChAlloc}. The Rx antenna polarization is $45^{\rm o}$ in both figures. Following the 5G NR standard, the Rx antenna polarization angle can also be analyzed based on the axes of the coordinates, $-45^{\rm o}$ and $45^{\rm o}$, and $\overline{\rm XPD}^{\rm Rx-Ant} = \tan^2 (45^{\rm o} - {\rm Rx~antenna~polarization~angle})$ in linear scale. The transmit power allocation ratio, $\alpha_{\rm XPD/XPR}$ and $\beta_{\rm XPD/XPR}$, is determined based on (\ref{eq:Optimal_alpha}) -- (\ref{eq:Optimal_beta}).
The scenario of PSI in Fig. \ref{fig:Statistic_SER_0_RxAnt} is $\overline{\rm XPD}^{\rm N} = 5.48$ dB, $\overline{\rm XPD}^{\rm P} = -6.26$ dB, $\overline{\rm XPR}^{\rm N} = 5.90$ dB. In contrast, The scenario of PSI in Fig. \ref{fig:Statistic_SER_0_RxAnt_WO_SubChAlloc} is $\overline{\rm XPD}^{\rm N} = 5.48$ dB, $\overline{\rm XPD}^{\rm N} = 4.56$ dB, $\overline{\rm XPD}^{\rm P} = -4.15$ dB, $\overline{\rm XPR}^{\rm N} = 4.34$ dB.  All those PSI parameters are practical and reasonable \cite{Landmann07, Kwon_Stuber11_TVT}.

\begin{figure}[!t]
	\centering
\includegraphics[width=.47\textwidth]{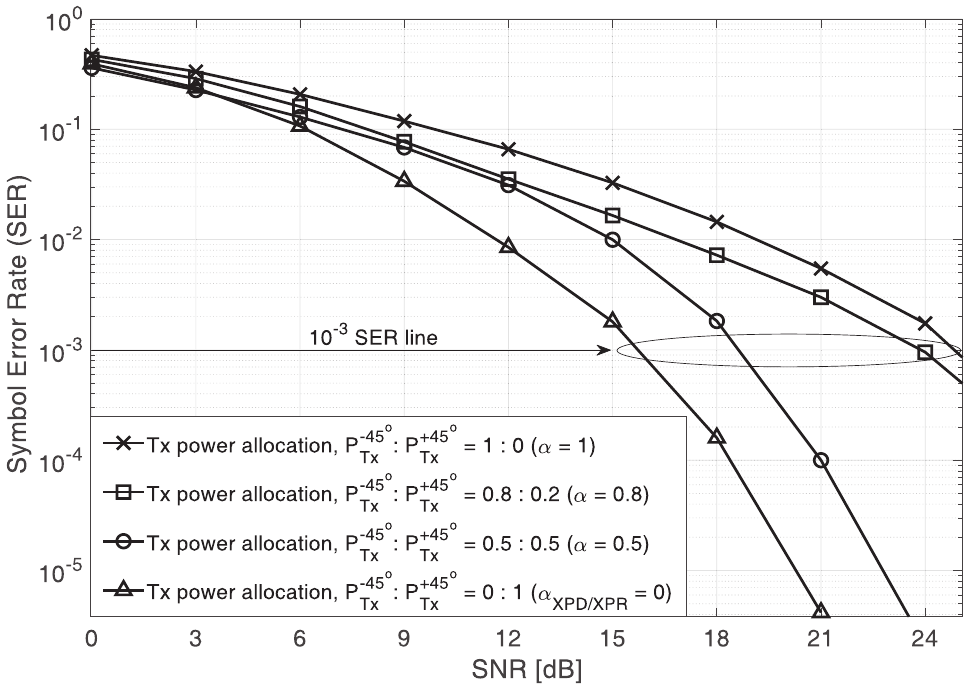}
	\caption{Symbol error rate for different transmit power allocation in the scenario of $45^{\rm o}$ Rx antenna polarization, $\overline{\rm XPD}^{\rm N} = 5.48$ dB, $\overline{\rm XPD}^{\rm P} = -6.26$ dB, $\overline{\rm XPR}^{\rm N} = 5.90$ dB.}
	\label{fig:Statistic_SER_0_RxAnt}
\end{figure}
\begin{figure}[!t]
	\centering
\includegraphics[width=.47\textwidth]{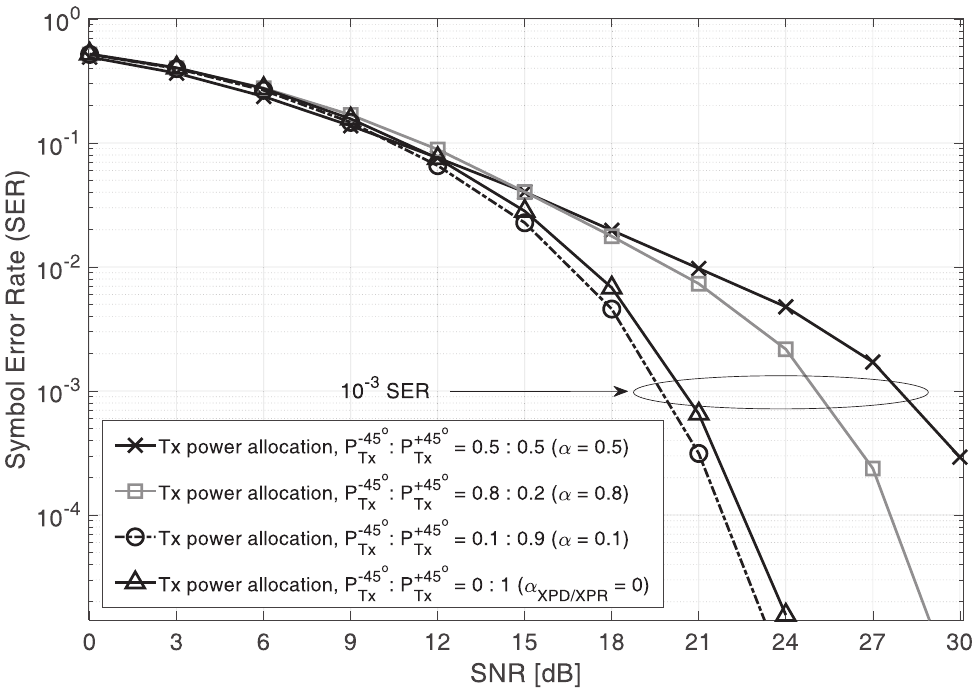}
	\caption{Symbol error rate for different transmit power allocation in the scenario of $45^{\rm o}$ Rx antenna polarization, $\overline{\rm XPD}^{\rm N} = 4.56$ dB, $\overline{\rm XPD}^{\rm P} = -4.15$ dB, $\overline{\rm XPR}^{\rm N} = 4.34$ dB.}
	\label{fig:Statistic_SER_0_RxAnt_WO_SubChAlloc}
\end{figure}

The simulation results apparently show that MPS-Beamforming with the XPD/XPR-aware transmit power allocation, remarkably improves the conventional beamforming based system performance in terms of SER. For instance, in the scenario of Fig. \ref{fig:Statistic_SER_0_RxAnt}, the novel proposed scheme with XPD/XPR-aware transmit power allocation ratio, $\alpha = 0$ and $\beta = 1$ has 9 dB gain of the SNR for $ 10^{-3}$ SER comparing with the scenario of $\alpha = 1$ and $\beta = 0$.  The different Rx antenna polarization and polarization state of the given channel require different transmit power allocation ratio as will be illustrated in this section. It is noteworthy that the selection of a $-45^{\rm o}$ polarization transmit beam achieves substantial SNR gain by 9 dB and 3 dB, respectively, comparing with the selection of a $45^{\rm o}$ polarization transmit beam or utilizing both $-45^{\rm o}$ and $45^{\rm o}$ polarization transmit beams with equal power allocation, in case of the channel state in Fig. \ref{fig:Statistic_SER_0_RxAnt}.

The theoretical XPD/XPR-aware transmit power allocation ratio, $\alpha_{\rm XPD/XPR}$ and $\beta_{\rm XPD/XPR}$ in (\ref{eq:Optimal_alpha}) -- (\ref{eq:Optimal_beta}), is obtained from the estimation of PSI in the given channel of Fig. \ref{fig:Statistic_SER_0_RxAnt_WO_SubChAlloc}. However, we emphasize that the proposed subcarrier assignment algorithm is not applied to the simulation in Fig. \ref{fig:Statistic_SER_0_RxAnt_WO_SubChAlloc}. $N = 48$ different subcarriers that have the highest channel gains are selected for MPS-Beamforming, i.e., PSI is not considered in subcarrier assignment. Even in this case, the simulation results show that MPS-Beamforming still remarkably improves the system performance in terms of SER. In Fig. \ref{fig:Statistic_SER_0_RxAnt_WO_SubChAlloc}, the proposed scheme of XPD/XPR-aware transmit power allocation with $\alpha_{\rm XPD/XPR} = 0$ and $\beta_{\rm XPD/XPR} = 1$ has 7 dB gain of the SNR for $ 10^{-3}$ SER comparing to the case of arbitrarily combining two transmit beams adopting $-45^{\rm o}$ and $45^{\rm o}$ polarization with $\alpha = 0.5$ and $\beta = 0.5$.

On the other hand, the disadvantage of excluding the proposed subcarrier assignment algorithm is clearly described in the simulation of Fig. 5.
One of the important observations in Fig. 5 is that $\alpha_{\rm XPD/XPR}$ and $\beta_{\rm XPD/XPR}$ obtained by the proposed XPD/XPR-aware transmit power allocation scheme do not exhibit the best SER curve, since the MPS-Beamforming in the OFDM system does not apply the pertinent subcarrier assignment algorithm proposed and described in Table \ref{Tbl:SubcarrierAlloc}. The MPS-Beamforming with $\alpha = 0.1$ and $\beta = 0.9$ results in the better SNR curve than the MPS-Beamforming with $\alpha_{\rm XPD/XPR}$ and $\beta_{\rm XPD/XPR}$ resulting from the proposed scheme, i.e., 0.7 dB SNR gain at $10^{-3}$ SER.

\begin{figure}[!t]
	\centering
\includegraphics[width=.50\textwidth]{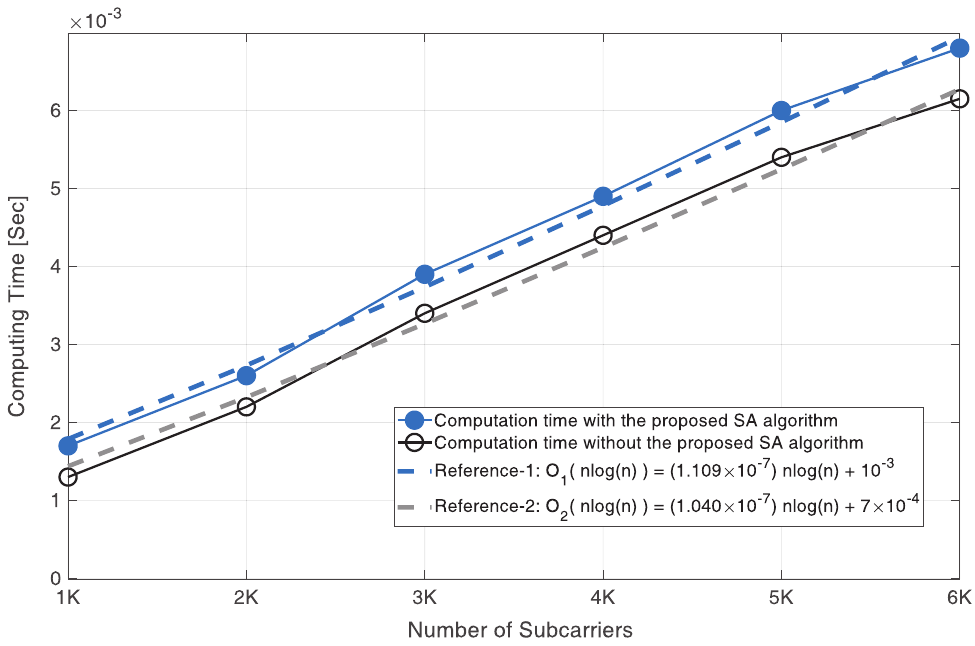}
	\caption{Complexity/computation time of the proposed subcarrier assignment (SA) algorithm for varying number of OFDM subcarriers.}
	\label{fig:Complexity_SubChAlloc}
\end{figure}

For the remaining simulation results with respect to the SER in Figs. \ref{fig:SER_-45_RxAnt} -- \ref{fig:SER_45_RxAnt_Statistical_Ch}, we emphasize that the proposed new subcarrier assignment algorithm in Table \ref{Tbl:SubcarrierAlloc} is adopted to pertinently support the XPD/XPR-aware transmit power allocation scheme and to result in the best SER performance of MPS-Beamforming in the OFDM system. The complexity of the proposed subcarrier assignment algorithm in terms of the computation time is illustrated in Fig. \ref{fig:Complexity_SubChAlloc}. The simulation results are in great agreement with the analysis of the complexity. The expected computation time for $n$ OFDM subcarriers, $T(n)$ is shown to be the Big-O function of $n \ln n$, i.e., \cite{Intro_to_Algorithms_Book}
\begin{equation}
 \label{eq:Complexity_SubChAlloc}
  T(n) = O(n \ln n).
\end{equation}
The computation with the proposed XPD-aware OFDM subcarrier assignment algorithm shows good agreement with Reference-1 in the legend of Fig. \ref{fig:Complexity_SubChAlloc}.

The simulation of the proposed subcarrier assignment algorithm utilizes MATLAB's sort function, which is based on Quicksort with the expectation of the complexity, $O(n \ln n)$. Considering the benefit of applying the proposed subcarrier assignment algorithm to be shown in the remaining simulation results, the cost in terms of the increased complexity is acceptable. The trade-off is the significant improvement in SNR, e.g., 4 dB SNR gain in statistical channel realization as will be presented in the remaining simulation results. The computation time without the proposed subcarrier assignment algorithm also follows Big-O function as described in Reference-2 in the legend of Fig. \ref{fig:Complexity_SubChAlloc}.
  Further, although the total number of subcarriers increases from 1K ($1024$) to 6K ($6 \times 1024$), the increase of the additional computation time is marginal.


The analytic results of MPS-Beamforming are provided in Figs. (\ref{fig:Eccentricity}) -- (\ref{fig:Rx_22.5}).
It is shown that the instantaneous Rx XPD affects the rotation angle of the Rx polarization ellipse at the UE. A variety of phase difference, $\Delta$ in (\ref{eq:r}) -- (\ref{eq:Delta}) along with varying ${\rm XPD} = (E^{-45^{\rm o}}/E^{+45^{\rm o}}) ^2$ are taken into account. Further, several Rx antenna polarization angles at the UE are considered to reflect the movement of the UE on the rotation of the Rx antenna orientation. Although some prior research defines XPD in slightly different manners, we follow the prevalent conventional definition of instantaneous or statistical XPD, which is described in Section \ref{sec:MPS-Beamforming}.

The squared eccentricity is a function of XPD as described in (\ref{eq:Eccentricity}). The curves of the squared eccentricity $\epsilon^2$ for varying XPD and a variety of phase difference $\Delta$ are illustrated in Fig.\ref{fig:Eccentricity}. Four scenarios with respect to $\Delta$ are considered to analyze the behavior of squared eccentricity for varying XPD. In the ideal scenario that $\Delta = 0$, the squared eccentricity is unity, i.e., $\epsilon^2 = 1$; it corresponds to linear polarization created by the superposition of two synchronized transmit beams. This scenario is omitted in Fig. \ref{fig:Eccentricity} since the behavior is straightforward from (\ref{eq:Eccentricity}). On the other hand, if the phase difference is greater than $45^{\rm o}$, squared eccentricity at XPD = 1 is significantly less than unity. Even for the large phase difference, the squared eccentricity can be greater than 0.9 yielding the narrow ellipse of polarization. As mentioned in Section \ref{sec:Pol_ellipse}, as $\epsilon ^2$ approaches unity, the polarization ellipse converges to a linear line, which is the desired scenario.

\begin{figure}[!t]
	\centering
\includegraphics[width=.467\textwidth]{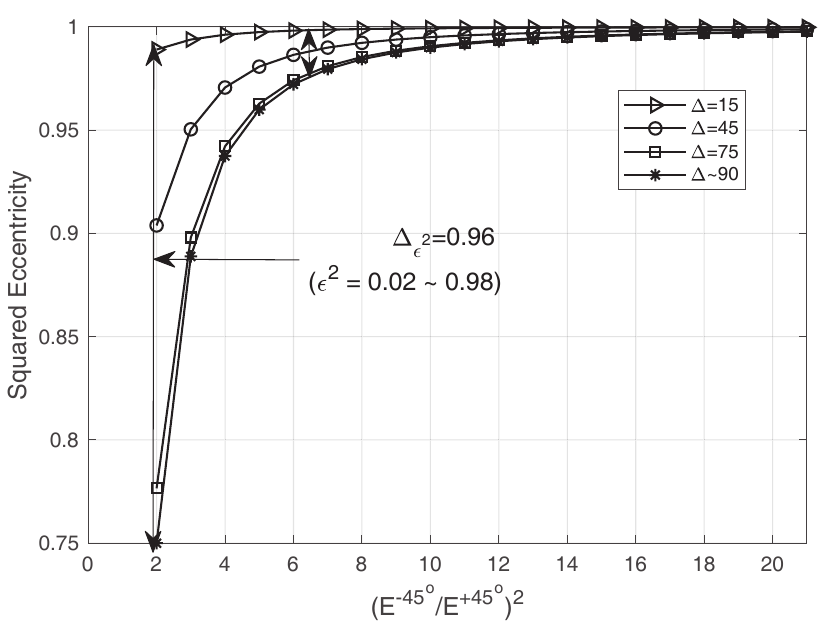}
	\caption{Squared eccentricity, $\epsilon ^2$ for the varying XPD.}
	\label{fig:Eccentricity}
\end{figure}
\begin{figure}[!t]
	\centering
\includegraphics[width=.457\textwidth]{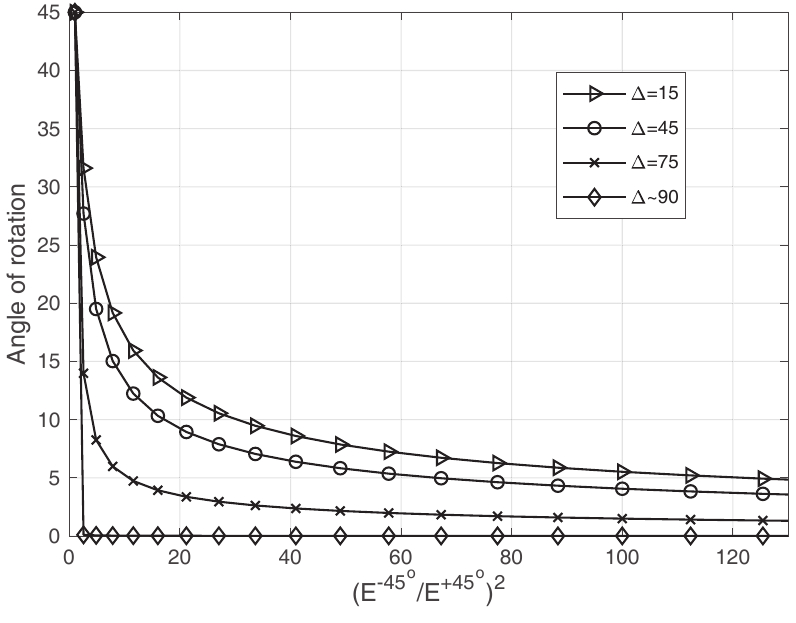}
	\caption{Angle of rotation, $\theta$ for the varying XPD.}
	\label{fig:Theta}	
\end{figure}
\begin{figure}[!t]
	\centering
\includegraphics[width=.457\textwidth]{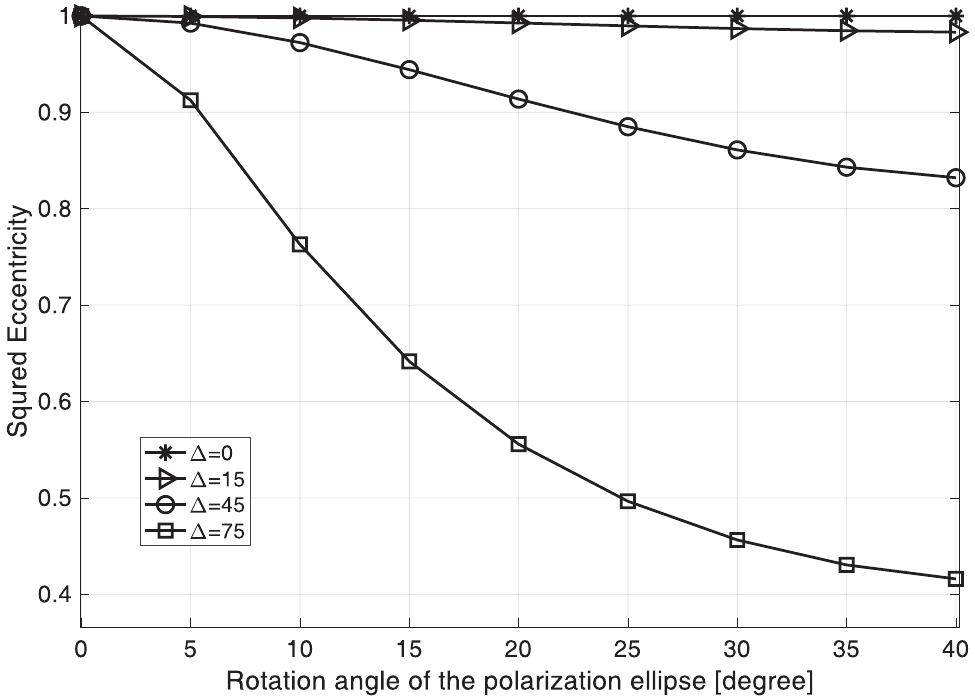}
	\caption{Direct relation between squared eccentricity, $\epsilon ^2$ and rotation angle of the polarization ellipse, $\theta$.}
	\label{fig:Eccentricity_Theta}	
\end{figure}

The polarization angle of the superimposed received signal at the UE, $\theta$, is a function of XPD as demonstrated in (\ref{eq:Theta}). The behavior of $\theta$ is depicted for varying XPD and a variety of phase difference in Fig. \ref{fig:Theta}. When XPD is unity in linear scale, it yields $\theta = 45^{\rm o}$ for all the scenarios with respect to the phase difference. In the coordinates described in Figs. \ref{fig:System_Model} and \ref{fig:Polarization_Ellipse}, $E^{-45^{\rm o}}$ is the component for $x$-axis; therefore, as XPD increases, the polarization on $x$-axis increases, corresponding to the scenario that $\theta$ converges to zero.

\begin{figure}[!t]
	\centering
\includegraphics[width=.47\textwidth]{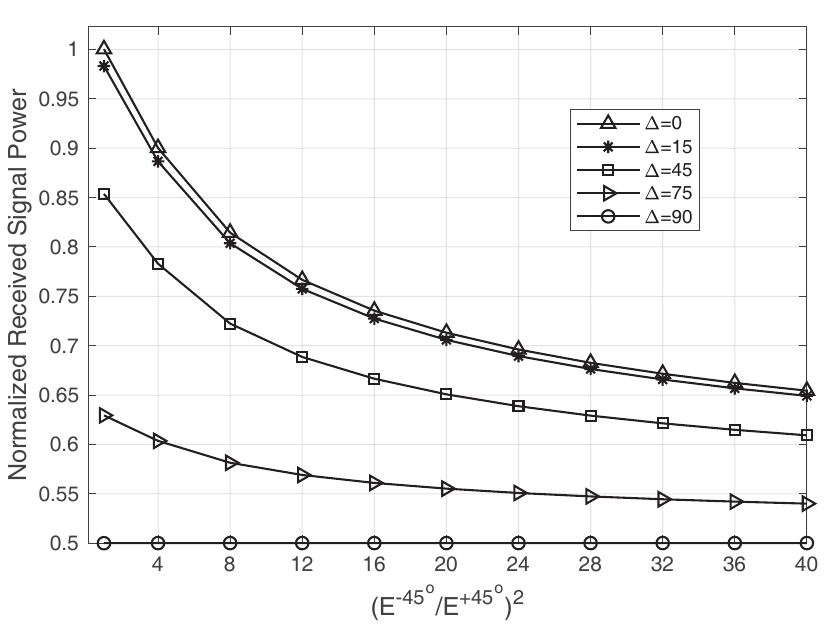}
	\caption{Normalized received signal power for $0^{\rm o}$ Rx antenna polarization.}
	\label{fig:Rx_45}
\end{figure}
\begin{figure}[ht]
	\centering
\includegraphics[width=.48\textwidth]{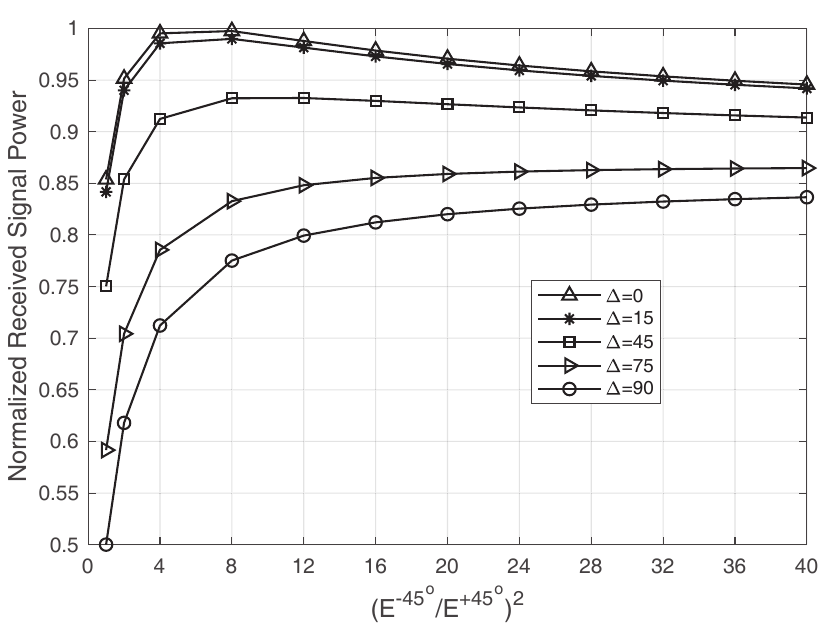}
	\caption{Normalized received signal power for $-22.5^{\rm o}$ Rx antenna polarization.}
	\label{fig:Rx_22.5}
\end{figure}


The direct relation between the squared eccentricity, $\epsilon ^2$ and rotation angle of the polarization ellipse, $\theta$ in (\ref{eq:Eccen(Theta)}) is portrayed in Fig. \ref{fig:Eccentricity_Theta}. As the polarization ellipse rotates further up to $45^{\rm o}$, the squared eccentricity tends to decrease, i.e., the polarization gets farther from the linear polarization, in particular, in the scenario that the phase difference $\Delta$ is greater than 45 such as $\Delta = 75^{\rm o}$. However, for the phase difference less than $45^{\rm o}$ ($\Delta \le 45^{\rm o}$), the rotation angle of the polarization ellipse do not affect the squared eccentricity seriously as exhibited in Fig. \ref{fig:Eccentricity_Theta}.

The next set of simulation results provided in Figs. \ref{fig:Rx_45} -- \ref{fig:Rx_22.5} depict the normalized received signal power for the limited total transmit power. XPD is varied from 1 to 40 in linear scale; further, different scenarios with diverse Rx antenna polarization angles of $0^{\rm o}$ and $-22.5^{\rm o}$ at the UE are taken into account.  It is noteworthy that the curves in Fig. \ref{fig:Rx_45} exhibit a different tendency from the curves in Fig. \ref{fig:Rx_22.5}, in terms of peak points. The reason for the different tendency is that different Rx antenna polarization requires different XPD of the received signal. Elaborate analyses are provided in the sequel.


The normalized received signal power at the Rx with $0^{\rm o}$ Rx antenna polarization is illustrated in Fig. \ref{fig:Rx_45}. In the scenario that $\Delta = 0^{\rm o}$, and ${\rm XPD} = 1$; the polarization angle of the received signal, $\theta$ is $45^{\rm o}$. 
That is, the polarization of the received signal is exactly aligned with the Rx antenna polarization; thus, the received signal power will be maximized. As illustrated in Fig. \ref{fig:Rx_45}, comparing to the scenario of no MPS-Beamforming, i.e., $E^{+45^{\rm o}} = 0$ or equivalently, XPD = $\infty$, MPS-Beamforming with XPD = 1 can save 35$\%$ of transmission power in coherent MPS-Beamforming ($\Delta = 0^{\rm o}$).
Nonetheless, when $\Delta < 45^{\rm o}$, the UE can still expect approximately 85$\%$ of the maximum received signal power. They are impressive results describing that the proposed MPS-Beamforming is significantly energy-efficient, i.e., we can minimize power loss when the polarization ellipse is well aligned to the Rx antenna polarization.


\begin{figure}[!t]
	\centering
\includegraphics[width=.47\textwidth]{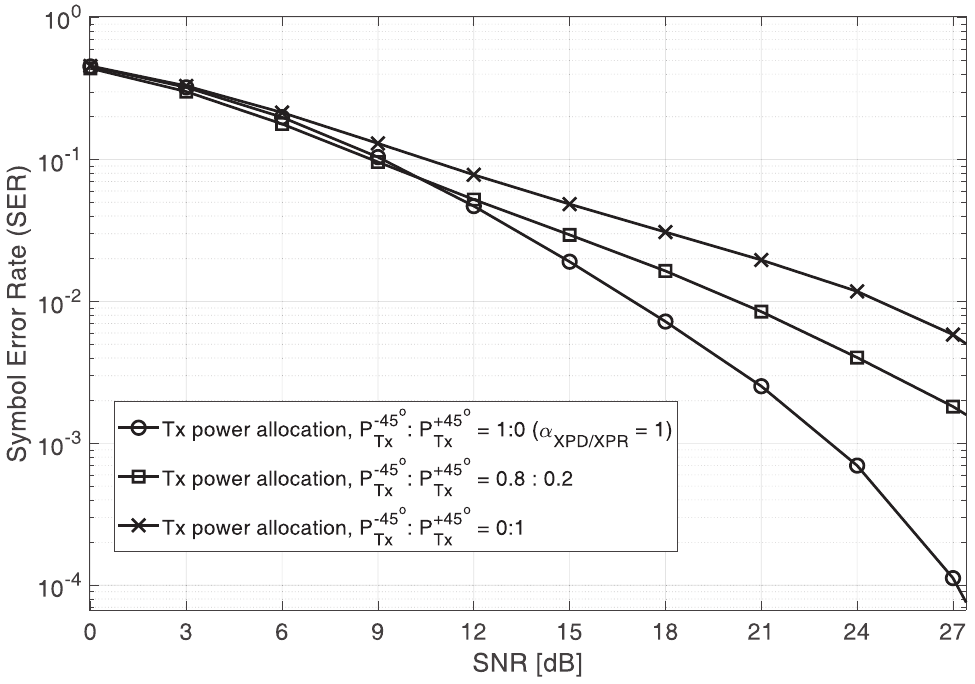}
	\caption{Symbol error rate for different transmit power allocation in the scenario of $-45^{\rm o}$ Rx antenna polarization.}
	\label{fig:SER_-45_RxAnt}
\end{figure}
\begin{figure}[!t]
	\centering
\includegraphics[width=.47\textwidth]{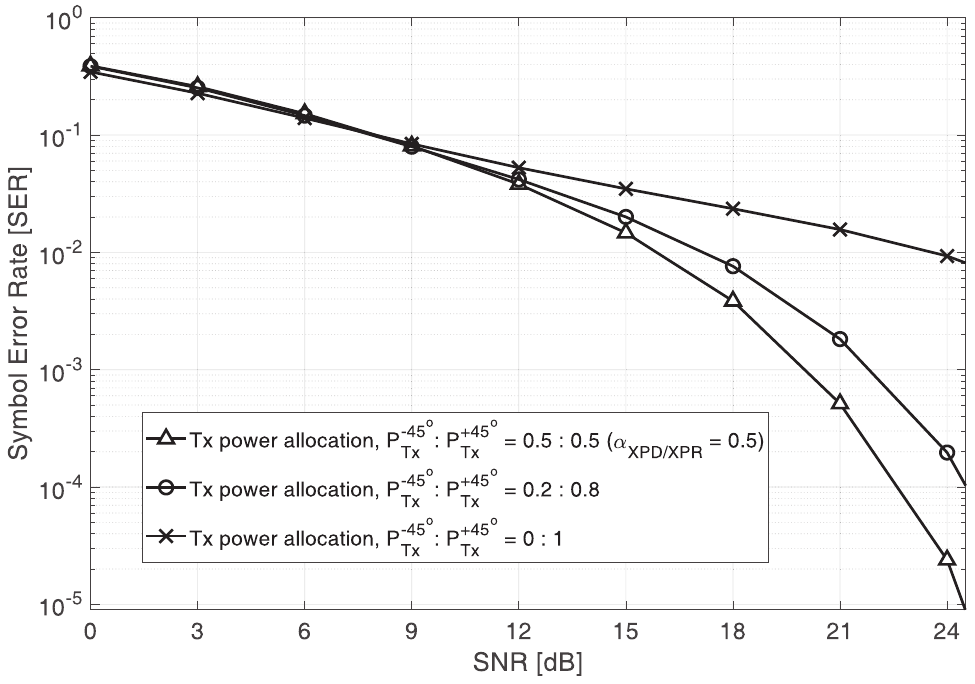}
	\caption{Symbol error rate for different transmit power allocation in the scenario of $0^{\rm o}$ Rx antenna polarization.}
	\label{fig:SER_0_RxAnt}
\end{figure}


The curves of normalized received signal power for the Rx antenna polarization angle, $-22.5^{\rm o}$ is depicted in Fig. \ref{fig:Rx_22.5}. It is noteworthy that the UE cannot have a peak for the normalized received signal power at the point, ${\rm XPD} = 1$. Eventually, it will be matched when ${\rm XPD} = \tan^2 \big( 45^{\rm o} - (-22.5^{\rm o}) \big) = 5.83$, which is in good agreement with Fig. \ref{fig:Rx_22.5}, where the curve with $\Delta = 0^{\rm o}$ has a peak of the normalized received signal power at ${\rm XPD} = 5.8$.

With the observation that different Rx antenna polarization requires different XPD of the received signal in Figs. \ref{fig:Rx_45} -- \ref{fig:Rx_22.5}, the different XPD/XPR-aware transmit power allocation ratios for different Rx antenna polarization are described and verified in Figs. \ref{fig:SER_-45_RxAnt} -- \ref{fig:SER_0_RxAnt}.
The Rx antenna polarization angles are $-45^{\rm o}$, and $0^{\rm o}$ in Figs. \ref{fig:SER_-45_RxAnt} and \ref{fig:SER_0_RxAnt}, respectively. The corresponding XPD/XPR-aware transmit power allocation ratios are, respectively, $\alpha_{\rm XPD/XPR} = 1$ and $\alpha_{\rm XPD/XPR} = 0.5$; it is noteworthy that $\alpha_{\rm XPD/XPR} = 1$ is the case of single transmit beam with $-45^{\rm o}$ polarization. In the scenario of $0^{\rm o}$ Rx antenna polarization in Fig. \ref{fig:SER_0_RxAnt}, the value of $\alpha$ which satisfies (\ref{eq:Pol_matching}) with (\ref{eq:XPD_MPS_Definition}) -- (\ref{eq:XPD_MPS_Fn_of_alpha}) is 0.5; i.e., $\alpha_{\rm XPD/XPR} = 0.5$ in (\ref{eq:Optimal_alpha}) -- (\ref{eq:Optimal_beta}) as described in the legend of Fig. \ref{fig:SER_0_RxAnt}.
In each scenario of Rx antenna polarization, the proposed XPD/XPR-aware transmit power allocation with $\alpha_{\rm XPD/XPR}$ and $\beta_{\rm XPD/XPR} = 1 - \alpha_{\rm XPD/XPR}$  in (\ref{eq:Optimal_alpha}) -- (\ref{eq:Optimal_beta}), outperforms that of any other choice of $\alpha$.

Finally, two sets of simulation results in both deterministic and statistical channels are presented in Figs. \ref{fig:SER_40_RxAnt_Deterministic_Ch} -- \ref{fig:SER_45_RxAnt_Statistical_Ch}. It is noteworthy that the best SER curve in each figure is resulting from the proposed subcarrier assignment algorithm along with the XPR/XPR-aware transmit power allocation scheme.
The PSI of the channels are, $\overline{\rm XPD}^{\rm N} = 4.56$ dB, $\overline{\rm XPD}^{\rm P} = -4.15$ dB, and $\overline{\rm XPR}^{\rm N} = 4.34$ dB, although the detailed channel impulse response of one channel realization is different from others. For the same PSI of the channel in a statistical sense, Figs. \ref{fig:SER_40_RxAnt_Deterministic_Ch} and \ref{fig:SER_40_RxAnt_Statistical_Ch} are in the scenario that the Rx antenna polarization is $5^{\rm o}$; whereas, the scenario of Figs. \ref{fig:SER_45_RxAnt_Deterministic_Ch} and \ref{fig:SER_45_RxAnt_Statistical_Ch} is $0^{\rm o}$ Rx antenna polarization.

The estimated $\alpha_{\rm XPD/XPR}$ based on the proposed XPD/XPR-aware transmit power allocation scheme, is 0.28 in the scenario of Figs. \ref{fig:SER_40_RxAnt_Deterministic_Ch} and \ref{fig:SER_40_RxAnt_Statistical_Ch}, and it is verified that the MPS-Beamforming with the XPD/XPR-aware transmit power allocation ratio, $\alpha_{\rm XPD/XPR}$ and the proposed subcarrier assignment algorithm results in the best SER curve in Figs. \ref{fig:SER_40_RxAnt_Deterministic_Ch} and \ref{fig:SER_40_RxAnt_Statistical_Ch}. In statistical channels of Fig. \ref{fig:SER_40_RxAnt_Statistical_Ch}, we can observe around 3 dB and 2.5 db SNR gain for $10^{-5}$ and $10^{-4}$ SER, respectively, in a statistical sense. That is, the proposed MPS-Beamforming with the novel transmit power allocation scheme and the new subcarrier assignment algorithm outperforms the single transmit beamforming with $-45^{\rm o}$ polarization with the highest SNR gain among a variety of $\alpha$'s. It is worth mentioning that in a given particular channel environment such as the deterministic channel in Fig. \ref{fig:SER_40_RxAnt_Deterministic_Ch}, the SNR gain can be higher than the statistical SNR gain. For instance, Fig. \ref{fig:SER_40_RxAnt_Deterministic_Ch} represents 6 dB and 5 dB SNR gain for $10^{-5}$ and $10^{-4}$ SER, respectively.

\begin{figure}[t]
	\centering
\includegraphics[width=.47\textwidth]{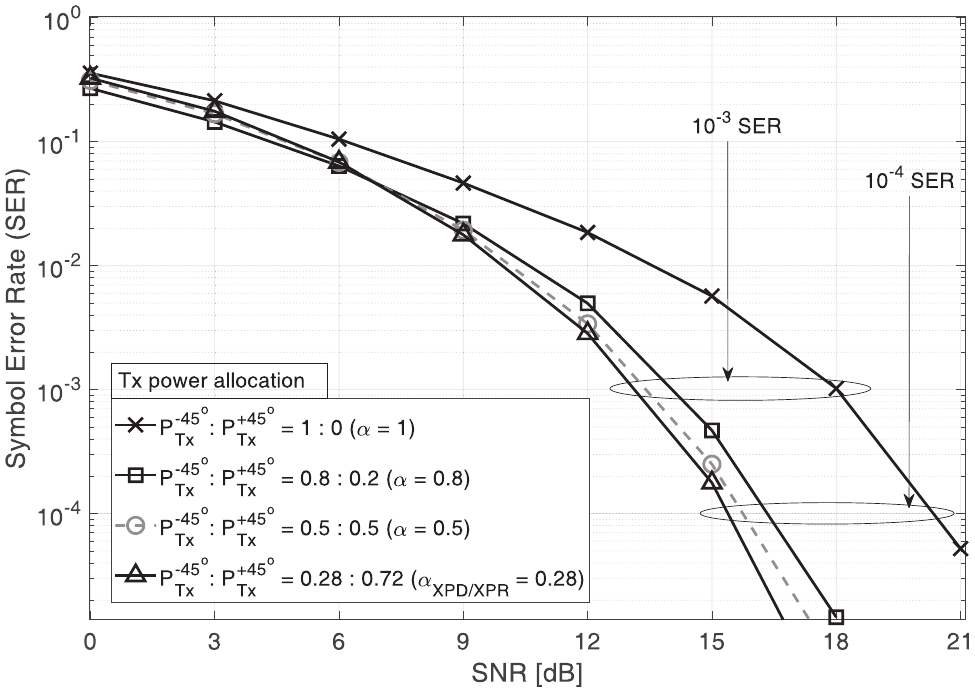}
	\caption{Verification of the proposed XPD/XPR-aware transmit power allocation and subcarrier assignment scheme showing the best symbol error rate in a deterministic channel; Rx antenna polarization is $5^{\rm o}$.}
	\label{fig:SER_40_RxAnt_Deterministic_Ch}
\end{figure}
\begin{figure}[!t]
	\centering
\includegraphics[width=.47\textwidth]{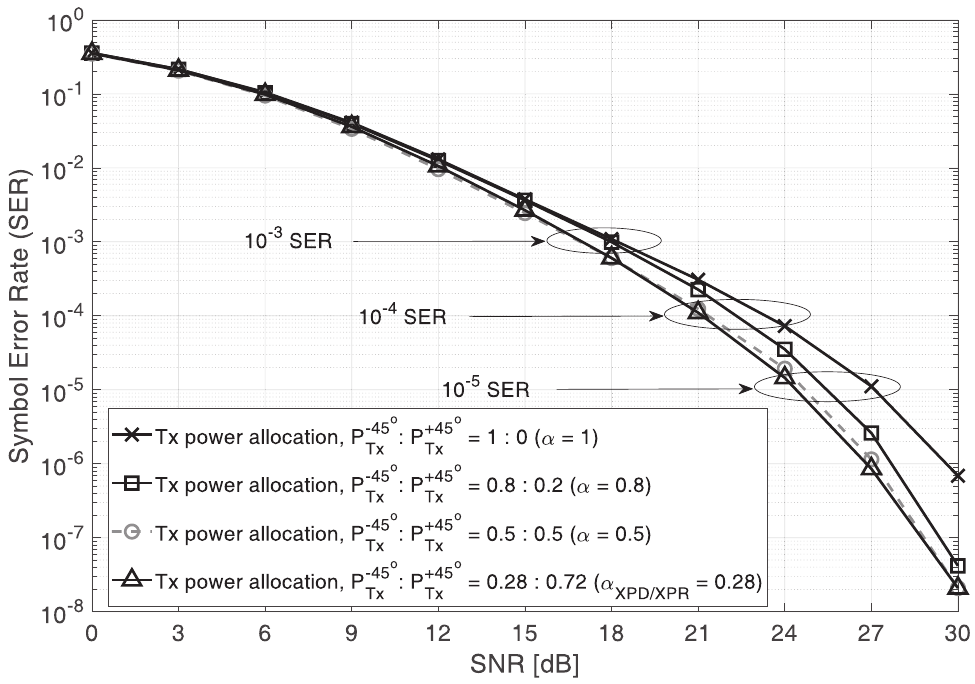}
	\caption{Verification of the proposed XPD/XPR-aware transmit power allocation and subcarrier assignment scheme showing the best symbol error rate in statistical channels; Rx antenna polarization is $5^{\rm o}$.}
	\label{fig:SER_40_RxAnt_Statistical_Ch}
\end{figure}

\begin{figure}[t]
	\centering
\includegraphics[width=.47\textwidth]{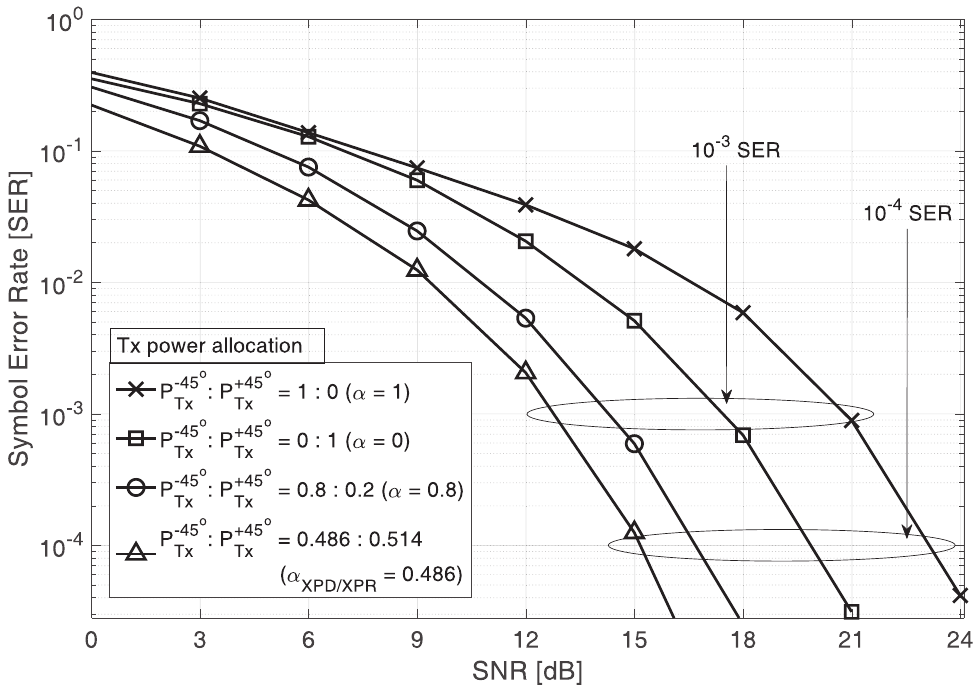}
	\caption{Verification of the proposed XPD/XPR-aware transmit power allocation and subcarrier assignment scheme showing the best symbol error rate in a deterministic channel; Rx antenna polarization is $0^{\rm o}$.}
	\label{fig:SER_45_RxAnt_Deterministic_Ch}
\end{figure}
\begin{figure}[t]
	\centering
\includegraphics[width=.47\textwidth]{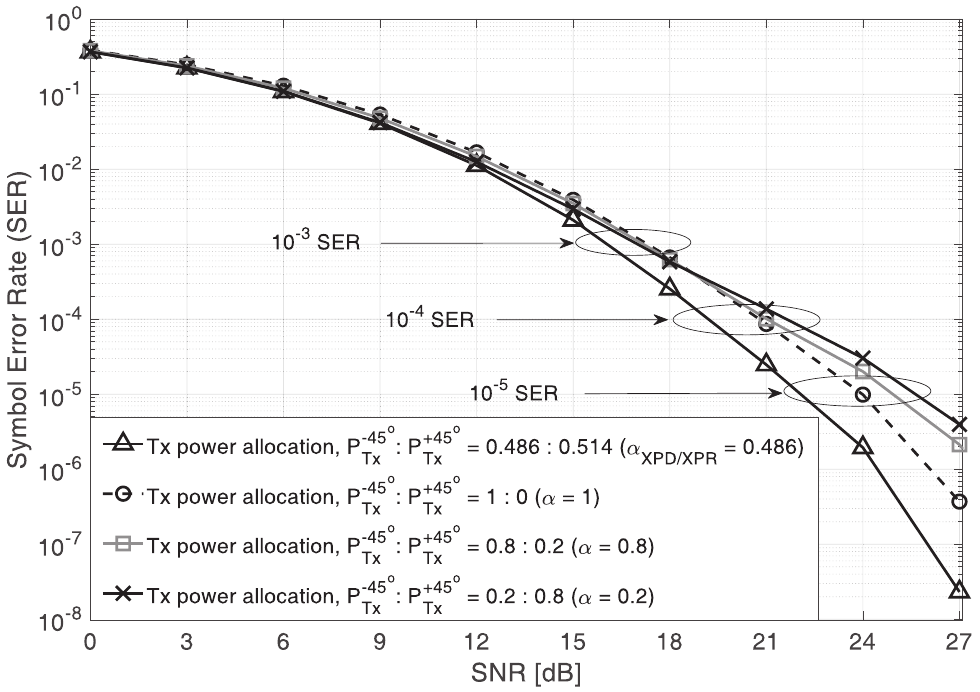}
	\caption{Verification of the proposed XPD/XPR-aware transmit power allocation and subcarrier assignment scheme showing the best symbol error rate in statistical channels; Rx antenna polarization is $0^{\rm o}$.}
	\label{fig:SER_45_RxAnt_Statistical_Ch}
\end{figure}

The Rx antenna can be rotated owing to the movement of the UE, and Figs. \ref{fig:SER_45_RxAnt_Deterministic_Ch} and \ref{fig:SER_45_RxAnt_Statistical_Ch} are in the scenario of $0^{\rm o}$ Rx antenna polarization. That is, the Rx antenna that has to receive the signal from the Tx is rotated by $5^{\rm o}$ from the scenario of Figs. \ref{fig:SER_40_RxAnt_Deterministic_Ch} and \ref{fig:SER_40_RxAnt_Statistical_Ch}. It is remarkable that the estimated $\alpha_{\rm XPD/XPR}$ based on the proposed XPD/XPR-aware transmit power allocation scheme, is substantially changed for the change of Rx antenna polarization by $5^{\rm o}$ at the aspect of the UE. The XPD/XPR-aware transmit power allocation ratio for the scenario of Figs. \ref{fig:SER_45_RxAnt_Deterministic_Ch} and \ref{fig:SER_45_RxAnt_Statistical_Ch} is $\alpha_{\rm XPD/XPR} = 0.486$; meanwhile, for Figs. \ref{fig:SER_40_RxAnt_Deterministic_Ch} and \ref{fig:SER_40_RxAnt_Statistical_Ch}, $\alpha = 0.28$, which is regarded as significant change. That is, as $\alpha_{\rm XPD/XPR}$ changes from 0.28 to 0.486, the MPS-Beamforming OFDM system imposes 74 $\%$ more transmission power on the transmit beamforming with $-45^{\rm o}$ polarization.

In the similar fashion with Figs. \ref{fig:SER_40_RxAnt_Deterministic_Ch} and \ref{fig:SER_40_RxAnt_Statistical_Ch}, the SNR gain of 4 dB and 2 db for $10^{-5}$ and $10^{-4}$ SER, respectively, are observed in a statistical sense, as illustrated by the simulation with statistcal channels in Fig. \ref{fig:SER_45_RxAnt_Statistical_Ch}. The proposed MPS-beamforming with the new subcarrier assignment algorithm along with the novel transmit power allocation scheme outperforms the single transmit beamforming with $-45^{\rm o}$ polarization, exhibiting the highest SNR gain among a variety of $\alpha$'s. It is noteworthy that a provided channel environment can show the more significant SNR gain such as around 8 dB for both $10^{-3}$ and $10^{-4}$ SER in Fig. \ref{fig:SER_45_RxAnt_Deterministic_Ch}.

\section{Conclusion} \label{sec:Conclusion}
This paper is the first to propose a novel scheme of MPS-Beamforming with XPD/XPR-aware transmit power allocation and the pertinent XPD/XPR-aware OFDM subcarrier assignment. Based on the 5G antenna panel structure agreed by the 5G NR standard society, the proposed scheme can be utilized to have the significant benefit of improving SER or SNR gain to satisfy the required SER; therefore, energy efficiency. The transmit power allocation ratio is theoretically obtained based on the given PSI of the wireless channel such as XPD and XPR. This theoretical scheme is verified by abundant simulations with a variety of scenarios in terms of PSI and Rx antenna polarization. Comprehensive simulation results show the remarkable improvement of the system; the noteworthy results include 8 dB SNR gain for $10^{-4}$ SER in a given realistic scenario, i.e., in a deterministic channel. Further, the long-term simulations with abundant statistical channels also exhibit the SNR gain of 4 dB for $10^{-5}$ in a provided realistic scenario of PSI and Rx antenna polarization in a statistical sense. The complexity of the proposed OFDM subcarrier assignment algorithm is acceptable considering the aforementioned benefits. The proposed scheme of MPS-Beamforming has the significant potential to be utilized in the advanced revision of 5G NR, beyond-5G or 6G wireless communication standards and systems in the future.

\section*{Acknowledgment}
The authors would like to thank the Editor and reviewers for their constructive comments that have helped improve the quality of this paper. Further, this work was supported financially by Aerospace Corporation (G255621100) and CSULB Foundation Fund (RS261-00181-10185); the authors appreciate the support.

\bibliographystyle{IEEEtran}
\bibliography{MPS_Beamforming_TWC_Rev36}

\begin{IEEEbiography}[{\includegraphics[width=1in,height=1.25in,clip,keepaspectratio]
{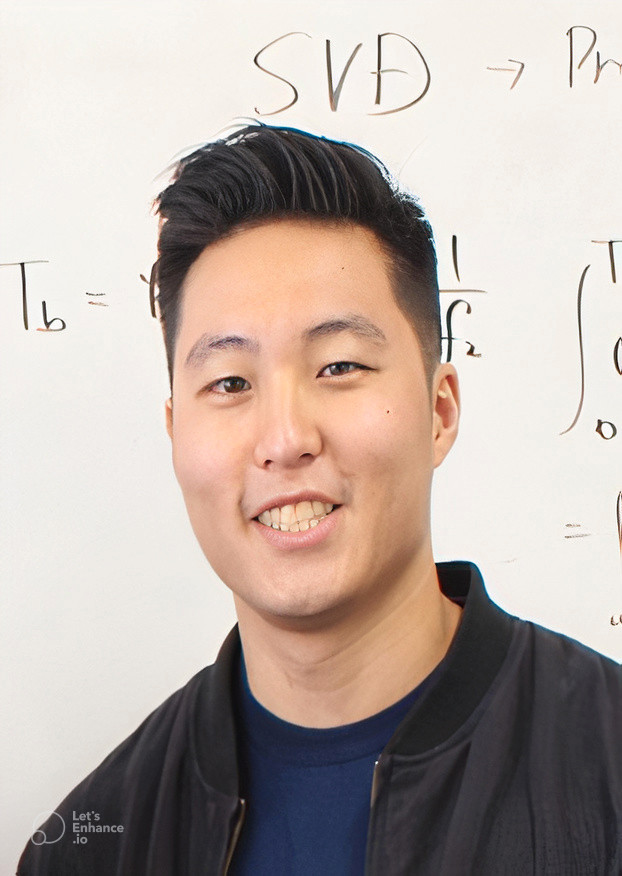}}]{Paul Seungcheol Oh} (S'19) received the B.Sc degree from California State University Long Beach with Magna Cum Laude in 2019. He is fulfilling research in the field of 5G and 6G Wireless Communications in Wireless Systems Evolution Laboratory (WiSE Lab) under the supervision of professor Sean Kwon, coauthor of this paper during the M.Sc program of California State University Long Beach.

He was a recipient of 6 term honors, President's list (three times) and Dean's list (three times). He received the First Place Award in Senior Capstone Design Team Project Competition of Electrical Engineering Department at California State University Long Beach, on May 2019, where the Panel of Judges is from Boeing, Aerospace Corporation, and Southern California Edison along with University Faculty. He was also a recipient of the Best Paper Award from IEEE Green Energy and Smart Systems conference (IGESSC), 2018.

\end{IEEEbiography}

\begin{IEEEbiography}[{\includegraphics[width=1in,height=1.25in,clip]
{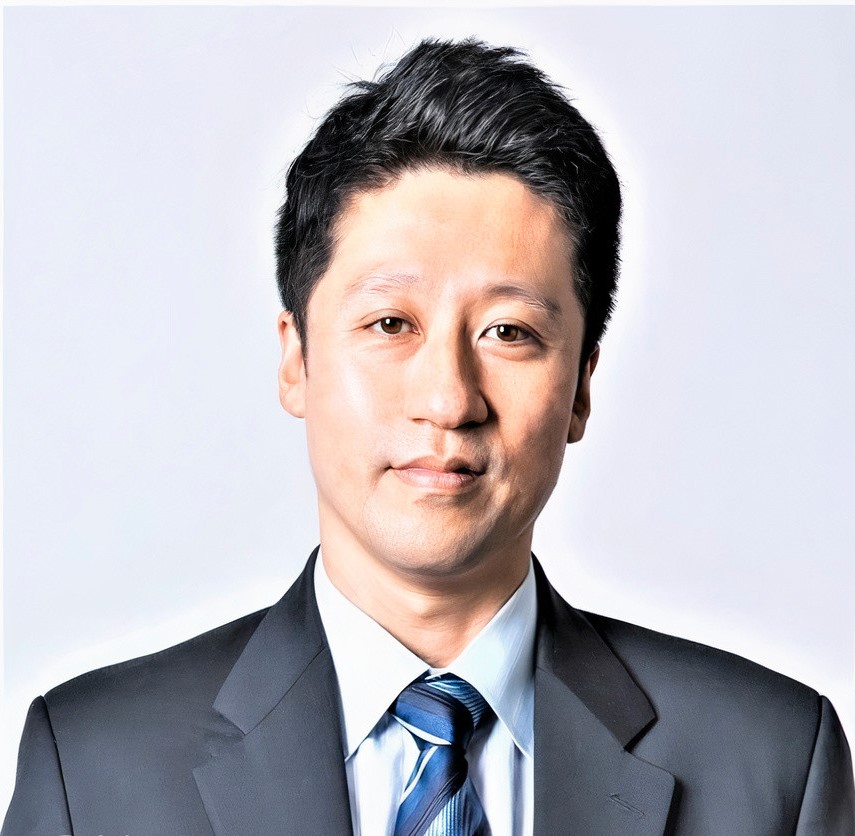}}]{Sean (Seok-Chul) Kwon} (S'09-M'14) received his Ph.D. degree from Georgia Institute of Technology in Atlanta, US on December 2013. Before that, Dr. Kwon received the M.Sc degree from the University of Southern California, Los Angeles, US in 2007; the B.Sc degree from Yonsei University, Seoul, South Korea in 2001.

He performed research at Intel's Next Generation and Standards Division in Communication and Devices Group during 2015 to 2017, where he contributed to 5G MIMO standards, the associated system design and patents. He has been Assistant Professor at California State University Long Beach; and Chief Director/Founder of Wireless Systems Evolution Laboratory (WiSE Lab) since 2017.

He also conducted postdoctoral research at Wireless Devices and Systems Group, University of Southern California in 2014 - 2015. He worked on CDMA common air interface focusing on layer-3 protocols at the R$\&$D Institute of Pantech co., Ltd, Seoul, South Korea in 2001 to 2004. He was involved in several projects such as a DARPA project; an US Army Research Lab project; and 6 mobile-station projects for Motorola and Sprint, which were successfully on the market. His current research interests are in 5G and beyond-5G wireless system/network design; satellite communications; polarization diversity and multiplexing; body area network such as wearable computing; wireless channel modeling and its applications; and application of machine learning for wireless communications and signal processing.
He was a recipient of 3 Best Paper Awards from IEEE Green Energy and Smart Systems Conference (IGESSC), 2018, 2019 and 2020.
\end{IEEEbiography}

\end{document}